\documentclass[letterpaper,12pt]{report}	
\usepackage[english,portuges,brazilian,brazil]{babel}
\usepackage[T1]{fontenc}
\usepackage[latin1]{inputenc}
\usepackage{amsthm,amsfonts,amssymb,amsmath}
\usepackage{graphicx}
\usepackage{makeidx}
\usepackage{fancyheadings}
\makeindex     
\newtheorem{thm}{Teorema}[section]
\newtheorem{cor}[thm]{Corolário}

\newtheorem{lem}[thm]{Lema}

\newtheorem{obs}[thm]{Observação}
\newtheorem{prop}[thm]{Proposição}
\newtheorem{defn}[thm]{Definição}
\theoremstyle{definition}
\newtheorem{exemp}[thm]{Exemplo}

\newcommand{\Z}{\mathbb Z}

\newcommand{\q}{\mathrm Q}
\newcommand{\F}{\mathbb F}

\newcommand{\av}{\mathcal O}              
\newcommand{\df}{\mathcal{D}_F}           
\newcommand{\Pf}{\mathcal{P}_F}           
\newcommand{\cf}{\mathcal{C}_F}           
\newcommand{\LL} {\mathcal{L}}            
\newcommand{\pf}{\mathbb{P}_F}            
\newcommand{\af}{\mathcal{A}_F}           
\newcommand{\proj}[1]{\mathbf{P}^{#1}}    
\newcommand{\X}{\mathcal{X}}              
\newcommand{\Y}{\mathcal{Y}}              
\newcommand{\Pponto}{\mathcal{P}}         
\newcommand{\N}{\mathbb N}

\newcommand{\x}{\times}

\newcommand{\implica}{\Rightarrow}
\newcommand{\volta}{\Leftarrow}
\newcommand{\emm}{\rightarrow}
\newcommand{\Emm}{\mapsto}
\newcommand{\sse}{\Leftrightarrow}

\newcommand{\bolaf}[1]{\mathbf{B}[#1]}
\newcommand{\uniao}{\cup}
\newcommand{\Uniao}{\bigcup}
\newcommand{\inter}{\cap}

\newcommand{\ptodo}{\forall}

\newcommand{\contido}{\subset}
\newcommand{\contem}{\supset}
\newcommand{\contidod}{\varsubsetneq}           
\newcommand{\barra}[1]{\overline{#1}}
\newcommand{\ba}[1]{\overline{#1}}
\newcommand{\modulo}[1]{\mathsf{mod}{#1}}
\newcommand{\tnormal}[1]{\mathsf{#1}}          
\newcommand{\congruo}{\equiv}
\newcommand{\menori}{\leqslant}
\newcommand{\maiori}{\geqslant}

\newcommand{\combinacao}[2]{{ #1 \choose #2}}

\newcommand{\ndivide}{\nmid}
\newcommand{\iso}{\simeq}
\newcommand{\funcao}[5]{\begin{array}{rccl}
                          #1 : & #2 & \rightarrow & #3 \\
			  & #4 & \mapsto & #5
		        \end{array} }

\newcommand{\ordem}{\prec}
\newcommand{\gera}[1]{\langle #1 \rangle}

\usepackage[active]{srcltx} 
\title{ Uma Construção Elementar de Códigos de Goppa Geométricos}
\author{N. M. S.}

\begin{document}
	\pagenumbering{roman} \setcounter{page}{0}
\thispagestyle{empty}

\begin{center}
{\Large Universidade Estadual de Campinas} \\ \vspace{1ex}
{\Large Instituto de Matemática, Estatística\\ e  Computação Científica} \\
 {\scshape
 {\large Departamento de Matemática}\\}
\vspace{1.0in}
{\Huge\textbf{Códigos Geométricos de Goppa Via Métodos Elementares}}\\
\vspace{1.0in} 
{\Large \textbf{Autor: Nolmar Melo de Souza}} \\ \vspace{1.0cm}
{\Large \textbf{Orientador: Prof. Dr. Paulo Roberto Brumatti}} \\ \vspace{2ex}
{\Large \textbf{Co-orientador: Prof. Dr. Fernando Eduardo Torres Orihuela}} 
\end{center}

	\thispagestyle{empty}

\hspace*{-1cm}
\vspace*{-1cm}
\includegraphics{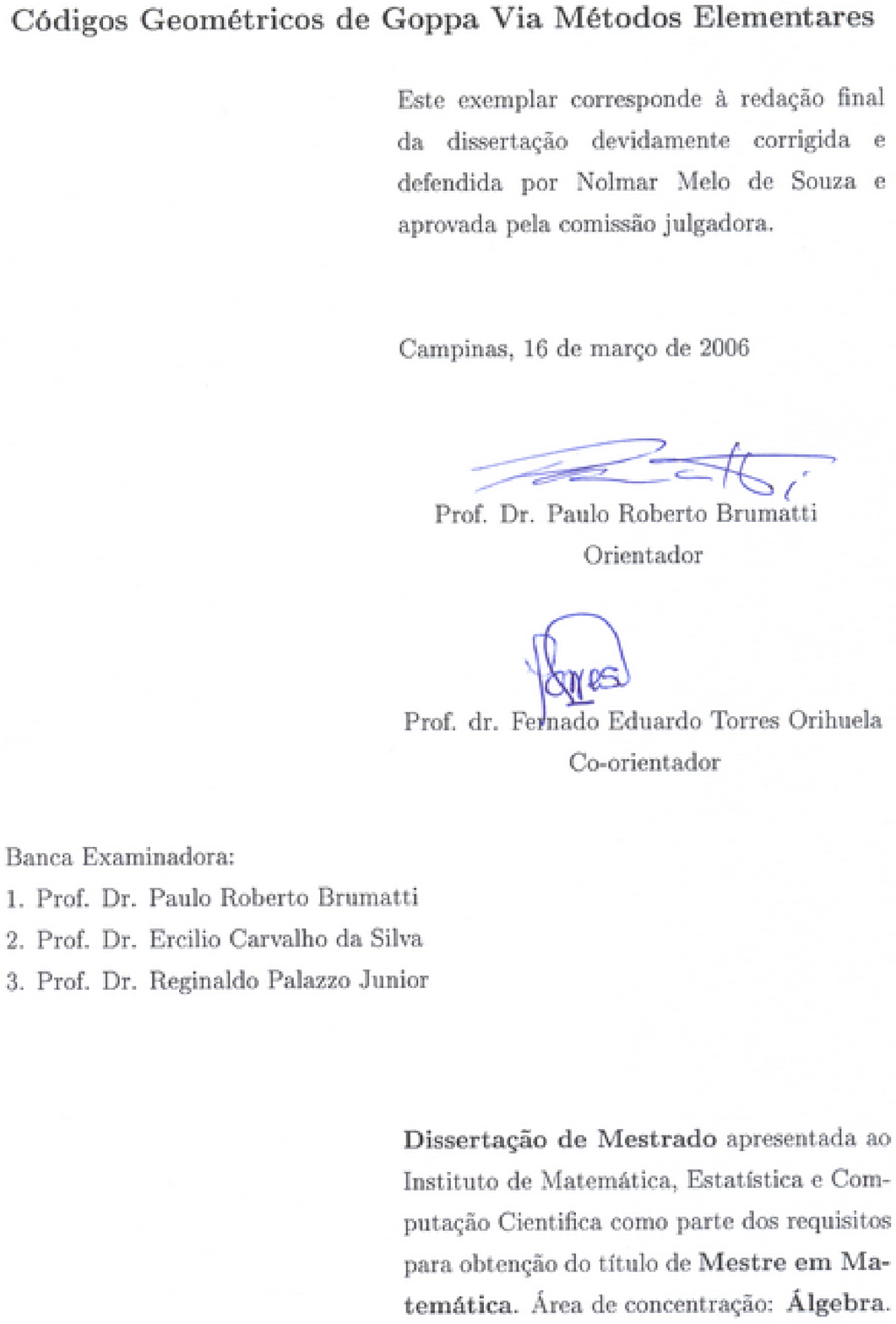}



\newpage \normalsize
\thispagestyle{empty}

\vspace{3.5cm}
\begin{center} 
\textbf{
\begin{tabular}{c} 
  FICHA CATALOGRÁFICA ELABORADA PELA \\
  BIBLIOTECA DO IMECC DA UNICAMP \\
  Bibliotecária: Maria Júlia Milani Rodrigues - CRB8a /2116\\
\end{tabular}
}
\end{center}
\begin{footnotesize}
\begin{center}
\begin{tabular}{|cl|} \hline
  \hspace{1cm} & \\
  	& Melo, Nolmar \\
  M491c & \hspace{0.6cm} Códigos geométricos de Goppa via métodos elementares / Nolmar Melo de\\
  	&  Souza -- Campinas, [S.P.:s.n.], 2006. \\
  	&  \\
  	& \hspace{0.6cm} Orientadores: Paulo Roberto Brumatti; Fernado Eduardo Torres Orihuela\\
	& \hspace{0.6cm} Dissertação (mestrado) - Universidade Estadual de Campinas, Instituto de \\
  	& Matemática, Estatística e  Computação Científica. \\
  	& \\
  	& \hspace{0.6cm} 1. Códigos de controle de erros (Teoria da informação). 2. Semigrupos. 3. \\
  	& Geometria algébrica. 4. Teoria da codificação. I. Brumatti, Paulo Roberto. II. \\
  	& Torres Orihuela, Fernando Eduardo. III. Universidade Estadual de Campinas. \\
  	& Instituto de Matemática, Estatística e Computação Científica. IV. Título \\
  	& \\ \hline
\end{tabular} 
\end{center}
\end{footnotesize}
\hspace*{-1cm}
\begin{tabular}{ll}
Título em inglês: Goppa goemetry codes via elementary methods  &\\ 
& \\
Palavras-chave em inglês (Keywords): 1. Error control codes (Information theory). 2.&\\
Semigroups. 3. Algebraic geometry. 4. Coding theory.& \\
& \\ 
Área de concentração: Álgebra &\\
& \\
Titulação: Mestre em Matemática &\\ 
& \\
\end{tabular}\\
\hspace*{-1cm}
\begin{tabular}{ll}
Banca examinadora: & Prof. Dr. Paulo Roberto Brumatti (IMECC-UNICAMP) \\
 & Prof. Dr. Ercílio Carvalho da Silva (UFU-MG) \\
 & Prof. Dr. Reginaldo Palazzo Junior (FEEC-UNICAMP)\\
\end{tabular}\\
\hspace*{-1cm}
\begin{tabular}{ll}
 & \\
Data da defesa: 17/02/2006 & \\
\end{tabular}

\thispagestyle{empty}
\begin{center}
\includegraphics{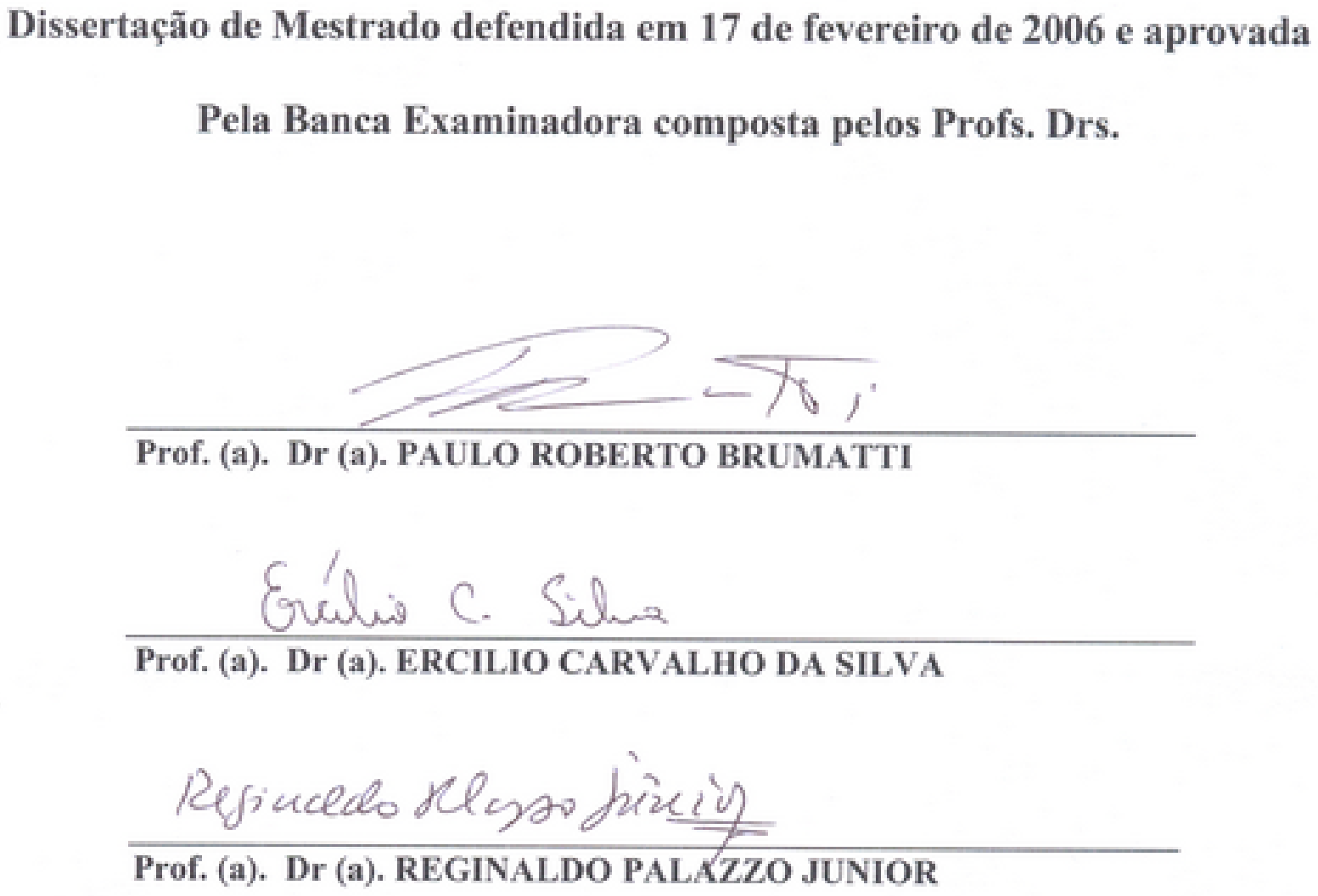}
\end{center}

	\newpage \thispagestyle{plain} 
\vspace{1.5cm}
\begin{center}
{\huge{\textbf{Resumo}}}
\end{center}
\vspace{0.5cm}


O objetivo central desta dissertação foi o de apresentar os Códigos Geométricos de Goppa via métodos elementares que foram introduzidos por J. H. van Lint, R. Pellikaan e T. H\o{}hold por volta de 1998. Numa primeira parte da dissertação são apresentados os conceitos fundamentais sobre corpos de funções racionais de uma curva algébrica na direção de se definir os códigos de Goppa de maneira clássica, neste estudo nos baseamos principalmente no livro ``Algebraic Function Fields and Codes'' de H. Stichtenoth. A segunda parte inicia-se com a introdução dos conceitos de funções peso, grau e ordem que são fundamentais para o estudo  dos Códigos de Goppa via métodos elementares de álgebra linear e de semigrupos, tal estudo foi baseado em ``Algebraic geometry codes'' de J. H. van Lint, R. Pellikaan e T. H\o{}hold. 

A dissertação termina com a apresentação de exemplos que ilustram os métodos elementares que nos referimos acima.

\newpage
\thispagestyle{plain} 
\vspace{1.5cm}
\begin{center}
{\huge{\textbf{Abstract}}}
\end{center}
\vspace{0.5cm}

The central objective of this dissertation was to present the Goppa Geometry Codes via elementary methods which were introduced by J.H. van Lint, R.Pellikaan and T. H\o{}hold  about 1998. On the first part of such dissertation are presented  the fundamental concepts about fields of rational functions of  an algebraic curve  in the direction as to define the Goppa Codes on a classical manner. In this study  we based ourselves  mainly on  the book ``Algebraic Function Fields and Codes'' of H. Stichtenoth. The second part is initiated with an introduction about the functions weight, degree and order which are fundamental for the study of the Goppa Codes through  elementary methods of linear algebra and of semigroups and such study was based on  ``Algebraic Geometry Codes'' of J.H. van Lint, R.Pellikaan and T. H\o{}hold.

The dissertation ends up with a presentation of examples which illustrate the elementary methods that we have referred to above.

\vspace{1.5ex}

\newpage \thispagestyle{plain}

	\thispagestyle{plain}
\vspace*{12cm}

\begin{flushright}
\begin{minipage}{12cm}
Dedico esse trabalho a aqueles que me apoiaram incondicionalmente, meus pais.
João Ferreira de Souza  e  Maria de Fátima Mello de Souza.
\end{minipage}
\end{flushright}

	\thispagestyle{plain}
\begin{Huge}
\centerline{{Agradecimentos\\}}
\end{Huge}
\vspace*{2.0cm}
Agradeço:
\begin{trivlist}  \itemsep 2ex  \normalsize

\item Por sua grande atenção, paciencia e dedicação agradeço em especial ao professor Paulo Roberto Brumatti, que me orientou na condução deste trabalho.

\item Ao professor Fernado Torres, o qual contribuiu com valiosas dicas de como prosseguir os estudos que levaram a construção deste.

\item Aos meus pais e irmãos que me apoiaram em todos os momentos da vida.

\item Ao professor Haroldo Benatti, que com sua dedicação me incentivou a continuar os estudos após a graduação.

\item À José Santana, que me apoiou em varios momentos.

\item Aos meus amigos, sem os quais seria dificil concluir essa etapa.

\item À Capes, pelo apoio financeiro.
\end{trivlist}

\newpage \normalsize
\thispagestyle{plain}

\vspace*{2.0in}
\hspace{0.6in}
{\sffamily\ttfamily
\begin{tabular}{|l|} \hline
\\
lance de dados\\
(gessinger)\\\\

daqui não tem mais volta, pra frente é sem saber\\
pequenos paraísos e riscos a correr\\
os deuses jogam pôquer\\
e bebem no saloon doses generosas de br 101\\
tá escrito há 6.000 anos em parachoques de caminhão\\
atalhos perigosos feito frases feitas\\
os deuses dão as cartas... o resto é com você\\
\\
no fundo tudo é ritmo\\
a dança foge do salão\\
invade a autoestrada do átomo ao caminhão\\
o fim é puro ritmo\\
o último suspiro é purificação\\
os deuses dão as costas... agora é só você\\
\\
os deuses dão as costas... agora é só você... querer 
\\
 
 \\ \hline
\end{tabular} 
}
\newpage{\empty}
	\tableofcontents
	\newpage
	\pagenumbering{arabic}
\markboth{Introdução}{Introdução}
\addcontentsline{toc}{chapter}{Introdução}
\chapter*{Introdução}

A teoria de códigos corretores de erros teve início com as pesquisas de matemáticos da Bell Lab. na década de 1940. Apesar da grande utilização na engenharia, a teoria utiliza de sofisticadas técnicas matemáticas fazendo uso de várias áreas tais como Geometria Algébrica, Teoria dos Números, Teoria dos Grupos e Combinatória.

Os códigos corretores de erros estão presentes no nosso cotidiano sempre que usamos o computador (no uso da internet ou na armazenagem de dados, por exemplo), transmitimos dados, assistimos um DVD, etc.

Os códigos algébricos geométricos, que são uma sub-classe dos códigos lineares (subespaços linear de um determinado espaço vetorial munido com uma métrica), foram inicialmente apresentados por V. D. Goppa num artigo, [3],  publicado em 1981. Essa classe de códigos, que é uma das mais estudadas atualmente, utiliza como ferramenta principal a geometria algébrica.

Nesse texto apresentamos os códigos de avaliação, que  foram introduzidos por 
Tom H\o{}hold , Jacobus H. van Lint e Ruud Pellikaan  através das funções ordem e peso. Tais códigos têm um tratamento muito mais simples do que  os códigos algébricos, geométricos visto que estes usam como teorias base os semigrupos e a álgebra linear.

O primeiro capítulo desse texto, capítulo onde a maioria das proposições tiveram suas demonstrações omitidas, traz uma introdução à teoria de códigos e também algumas ferramentas da geométria algébrica tendo como um dos principais resultados o teorema de Riemann-Roch.

Já no segundo capítulo fazemos uma apresentação dos códigos algébricos geométicos, códigos que também são chamados de {\bf códigos geométricos de Goppa}. Tal apresentação é feita de maneira sucinta.

No terceiro capítulo trazemos as funções ordem, peso e grau, que são fundamentais para calcularmos os parâmetros dos códigos de avaliação, que alí também são apresentados. Tais funções, unidas à teoria de semigrupo, nos permite descrever os códigos de avaliação de modo simples, fazendo desses uma alternativa aos códigos algébricos geométricos.

No quarto e último capítulo fazemos a conexão entre códigos de Goppa pontuais e códigos de avaliação e apresentamos alguns exemplos. Alí tentamos mostrar de modo prático a diferença entre ambos os códigos e ao mesmo tempo a proximidade dos dois.


\chapter{Noções Básicas}
Este capítulo está dedicado a introdução da notação que utilizaremos, assim como   expor definições e teoremas clássicos da teoria, os quais, na maioria das vezes, terão suas demonstrações omitidas, visto que a bibliografia indicada ao final desse texto as fazem muito bem. Um dos principais teoremas que iremos encontrar aqui é o Teorema de Riemann-Roch.


\section{Códigos}
Nessa seção iremos introduzir um dos principais objetos que serão tratados nesse texto (código). Usaremos aqui um conjunto $Q$ com $q$ elementos e o chamaremos de alfabeto.


\begin{defn}\index{Código}
	A um subconjunto próprio não vazio, $C$, de $Q^n$, damos o nome de código. Chamaremos de palavras código de comprimento $n$ aos seus elementos.
\end{defn}

\indent Dizemos que um código $C$ é trivial se $\#C=1$.

Em seguida apresentamos as definições elementares e o resultado, proposição \ref{codi:perf}, que caracteriza os códigos perfeitos.

\begin{defn}[Distância de Hamming] \index{Distancia! de Hamming}
	Sendo $x, y \in \q^n$ definimos distância de $x$ à $y$ como:
	$$d(x,y) = \#\lbrace i ; 1\menori i \menori n \textrm{ e } x_i \neq y_i \rbrace .$$
\end{defn}
\begin{defn}\index{Peso}
	Quando $Q$ é um corpo e $x\in Q^n$ o peso de $x$ é definido como $w(x)=d(x,0)$.
\end{defn}

\begin{defn}
	Chamamos de distância mínima do código $C$ ao natural $d=min\{d(x,y);\,\, x,y \in C\}$.
\end{defn}

\begin{prop}\label{codi:perf}
	Dado um código com distância mínima $d=2e+1$ temos que as bolas do conjunto $B = \lbrace \bolaf{x,e} ; x \in C \rbrace $ são disjuntas, onde $\bolaf{x,e}=\{y\in Q^n; d(x,y)\menori e\}$.
\end{prop}
\begin{proof}
	Segue diretamente da desigualdade triangular visto que a distân\-cia de Hamming é uma métrica.
\end{proof}
\begin{obs}
	Dizemos que um código com distância mínima $d=2e+1$ detecta $2e$ e corrige $e$ erros, ou seja, ocorrendo até $e$ erros é possível decodificar a palavra código transmitida.
\end{obs}

\begin{defn}
	Um código $C\contido \q^n$ com distância mínima $2e+1$ é dito código perfeito se $\q^n = \stackrel{\bullet}{\Uniao}_{x\in C} B(x,e)$.
\end{defn}
\section{Códigos Lineares} 
Nessa seção definiremos uma das principais classes de códigos e para tais códigos o  alfabeto, $\q$, é um corpo finito $\F_q$ com $q$ elementos, onde $q=p^r$ e $p$ é um número primo.
\begin{defn}[Código Linear] \index{Código! Linear}
	Um código linear $C$ sobre um alfabeto $\F_q$ é um $\F_q-$subespaço vetorial de $\F_q^n$. Se $dim_{F_q}(C) = k$ chamamos $C$ de um $[n,k]$-código.
\end{defn}
Usaremos a notação $[n,k,d]$-código, para um código linear com distância mínima $d$.
\begin{thm}
	Num código linear a distância mínima é igual ao peso mínimo.
\end{thm}
\begin{proof}
	Como $C$ é um subespaço vetorial e dados $x,y\in C$ temos que $x-y \in C$. Sejam $x,y \in C$ tais que $d(x,y)$ seja mínima. Logo $d(x,y) = d(x-y,0) = w(x-y)$ e mais $w(x-y)$ é mínimo.
\end{proof}

Na descrição de um código linear um importante ingrediente é o que se chama {\bf matriz geradora}.
\begin{defn}\index{Matriz! geradora}
	Uma matriz $k\x n$, $G$, é dita geradora de um código linear $C$ se os seus vetores linha formam uma base para $C$.   
\end{defn}

\begin{defn}
	Dizemos que uma matriz geradora, $G$, de um código linear, $C$, está na forma padrão se esta matriz está escrita em blocos da seguinte maneira:
	$$G= [ I_k|P ]  $$
	onde $I_k$ é a matriz identidade $k\x k$ e $P$ é uma $k\x (n-k)$ matriz.
\end{defn}
Em \cite{hefez2} vemos que nem todo código tem uma matriz geradora na forma padrão, contudo vê-se também que pode-se definir {\bf códigos equivalentes} com os mesmos parâmetros $n,\,k$ e $d$ de modo que essa possua uma matriz geradora na forma padrão.
\begin{prop}
	O complemento ortogonal de um $[n,k,d]-$código $C$, em rela\-ção ao produto interno canônico\footnote{Dados dois vetores em $\F_q^n$, $a=(a_1,\ldots,a_n)$ e $b=(b_1,\ldots,b_n)$, definimos o produto interno canônico como $\gera{a,b}=\sum_{i=1}^n a_ib_i$.}, também é um código e esse é chamado código dual e denotado por $C^{\bot}$. Além disso a dimensão do código dual é $n-k$ e mais, sendo $G=[I_k|P]$ a matriz geradora do código $C$, então $H = [-P^{t}|I_{n-k}]$ é a matriz geradora do código dual $C^{\bot}$. A matriz $H$ assim definida é dita matriz de teste de paridade pois $x \in C$ se, e somente se, $Hx^t = 0$.
\end{prop}
A proposição acima nos dá uma informação muito importante sobre a pertinência de uma dada palavra ao código, o que é fundamental para a sua decodificação. Assim ficam justificadas a definição e o resultado que apresentaremos a seguir:
\begin{defn}\index{Sindrome}
	Sendo $C$ um código linear com matriz de teste de paridade $H$, então para todo $x\in \F_q^n$ chamamos $Hx^t$ de síndrome de $x$.  
\end{defn}
\begin{thm}
	Dados $x,y\in \F_q^n$ então $x$ e $y$ tem a mesma síndrome se, e somente se, $x-y \in C$.
\end{thm}
\begin{proof}
	Seja $H$ a matriz de teste de paridade do código $C$, então temos:
	  $$Hx^{t} = Hy^{t} \sse Hx^{t}-Hy^{t} = 0 \sse H(x-y)^{t} = 0 \sse x-y \in C.$$
\end{proof}

As definições e resultados que apresentaremos abaixo estão todos relacionados com a codificação de uma mensagem enviada.

\begin{defn}
	Seja $c$ uma palavra transmitida e $x$ o vetor recebido, definimos o vetor erro como sendo $e=x-c$, observamos que a quantidade de erros ocorridos na transmissão é o peso de $e$.
\end{defn}
\begin{defn} \index{Classe! Lateral}
	Dizemos que dois vetores estão numa mesma classe lateral se eles tem a mesma síndrome. Um vetor de menor peso na classe é dito líder.
\end{defn}

\begin{thm}
	Sendo $C$ um $[n,k,d]-$código, se $u\in \F_q^n$ é tal que $w(u)\menori \left[ \frac{d-1}{2}\right]$\footnote{A notação $\left[ \frac ab \right]$ denota o maior inteiro menor ou igual a $\frac ab$.} então $u$ é o líder de sua classe.
\end{thm}
\begin{proof}
	Sejam $u,v\in \F_q^n$, com $w(u)\menori \left[ \frac{d-1}{2}\right]$ e $w(v) \menori \left[ \frac{d-1}{2}\right]$, se $u-v \in C$ temos que $w(u-v)\menori w(u)+w(v) \menori \left[ \frac{d-1}{2}\right]+\left[ \frac{d-1}{2}\right] \menori d-1$, logo $u-v=0$, assim temos que    $u=v.$
\end{proof}
\begin{thm}\label{teo:xx}
	Sendo $H$ uma matriz de teste de paridade de $C$ temos que $w(C) \maiori r$ se, e somente se, quaisquer $r-1$ colunas de $H$ são lineramente independentes (L.I.).
\end{thm}
\begin{proof}
	$\volta)$ Seja $c\in C$, $c\neq 0$, $H=[h_1,h_2,\cdots,h_n]$. Então $ 0=Hc^{t}=\sum_{i=1}^n c_ih_i$. Logo $w(c)\maiori r$, pois, caso contrário teríamos uma combinação linear de $r-1$ vetores L.I. dando zero.\\
	$\implica)$ Suponha que exista $r-1$ colunas L.D., sem perda de generalidade podemos supor que são as $r-1$ primeiras colunas $h_1,\cdots h_{r-1}$. Logo existem $c_1,\cdots c_{r-1} \in \F_q$, não todos nulos, com $\sum_{i=1}^{r-1} c_ih_i =0$. Assim $c=(c_1,\cdots,c_{r-1},0,\cdots,0)\in C$ pois $Hc^{t}=0$ e $w(c)\menori r-1$, absurdo.
\end{proof}
Uma cota que relaciona os parâmetros de um código aparecem no teorema:
\begin{thm}[Cota de Singleton] \index{Cota ! de Singleton}
	Seja $C$ um $[n,k,d]-$código linear, então $d \menori n-k+1$. Chamamos de código MDS (Maximum Distance Separable) ao código tal que $d=n-k+1$.
\end{thm}
\begin{proof}
	Seja $H$ a matriz de teste de paridade de $C$. Logo, posto de $H=n-k$. Do teorema \ref{teo:xx} segue que quaisquer $(d-1)$ colunas de $H$ são L.I.. Assim $d-1 \menori$ (posto de $H$), temos que $d\menori n-k+1$.
\end{proof}


\section{Lugares}
Nessa seção apresentaremos as estruturas necessárias para a definição dos códigos algébricos geométricos que veremos no próximo capítulo.
\begin{defn}
	\label{corpo_funcao} 
	Um corpo de funções algébricas\index{Corpo! de funções algébricas} $F/K$ de uma variável sobre $K$ é uma extensão de corpos $F \contem K$ tal que $F$ é uma extensão algébrica finita de $K(x)$ com $x\in F$ e transcendente sobre $K$.
\end{defn}

O conjunto $\Tilde{K} = \{a \in F ; a$ é algébrico em $K\} $ é um subcorpo de $F$ e mais $F/\Tilde{K}$ é um corpo de funções algébricas sobre $\Tilde{K}$. Esse conjunto é chamado de {\bf corpo de constantes}\index{Corpo! de constantes} de $F/K$. Neste trabalho vamos supor sempre que $K=\Tilde{K}$.

\begin{exemp}[Corpo de Funções racionais]
	O corpo de funções\index{Corpo! de funções racionais} $F/K$ é dito racional se $F=K(x)$ para algum $x$ que é transcendente sobre $K$, onde $K(x)$ é o corpo de frações do anel de polinômios, $K[x]$, em uma variável sobre o corpo $K$.
\end{exemp}
Qualquer elemento não nulo, $z\in K(x)$,  tem uma única representação 
$$z=a\prod_iP_i(x)^{n_i}$$ 
com $0\neq a \in K$, $P_i(x) \in K[x]$ mônicos irredutíveis distintos e $n_i \in \Z$.

A próxima definição nos traz um tipo de anel fundamental para o desenvolvimento do nosso trabalho.

\begin{defn}
    	Um anel de valorização\index{Anel! de valorização}  de um corpo de funções $F/K$ é um sub-anel $\av \contido F$ com as seguintes propriedades:
	\begin{enumerate}
    		\item $K \contidod \av \contidod F$;
		\item Para qualquer $z \in F$, $z \in \av$ ou $z^{-1} \in \av$.
	\end{enumerate}
\end{defn}
\begin{exemp}
	No corpo de funções racionais $K(x)$ aparecem os primeiros anéis de valorizações fundamentais, a saber:

	Seja $p(x) \in K[x]$, um polinômio irredutível. Definimos o conjunto
	$$\av_{p(x)} = \left\lbrace  \dfrac{f(x)}{g(x)} ; f(x),g(x) \in K[x], \,\, p(x)\ndivide g(x) \right\rbrace. $$
	Assim definido, $\av_{p(x)}$ é um anel de valorização de $K(x)/K$. Observe que se $q(x)$ for outro polinômio irredutível, não associado a $p(x)$, temos que $\av_{p(x)} \neq \av_{q(x)}.$
\end{exemp}

Os anéis de valorizações são caracterizados, segundo sua estrutura, pelos resultados:
\begin{prop}
    	Seja $\av$  um anel de valorização do corpo de funções $F/K$ então:
    	\begin{enumerate}
        	\item[a.] $\av$ é anel local, isto é $\av$ tem um único ideal maximal $P=\av\setminus \av^*$, onde $\av^*$ é o conjunto das unidades de $\av$;
        	\item[b.] Para $0\neq x \in F, \, \, x\in P \sse x^{-1} \notin \av$;
        	\item[c.] Para o corpo de constantes $\Tilde{K}$ de $F/K$ temos que $\Tilde{K} \contido \av$ e $\Tilde{K} \inter P = {0}.$
	\end{enumerate}
\end{prop}

Do item (b), segue que $\av$ é unicamente determinado por $P$. De fato, $\av_P:=\av=\{x\in F; x\notin P\}.$

\begin{thm}
    	Seja $\av$ um anel de valorização do corpo de funções $F/K$ e $P$ o seu ideal maximal. Então,
    	\begin{enumerate}
        	\item[a.] $P$ é principal;
        	\item[b.] Se $P= t\av$ então qualquer $0\neq z \in F$ tem uma única representação da forma $z=t^nu$, para algum $n\in \Z$ e $u \in \av^*$;
        	\item [c.] $\av$ é um domínio principal.
	\end{enumerate}
\end{thm}

Os anéis de valorarização nos levam a um outro conceito, a saber, o conceito de lugar.
\begin{defn}
    	\begin{enumerate}
        	\item Um lugar\index{Lugar} $P$ do corpo de funções $F/K$ é um ideal maximal de algum anel de valorização de $F/K$. Qualquer elemento $t$ que gera $P$ ($P=t\av$) é dito elemento principal de $P$;
        	\item $\pf=\{ P; \, \, P$ é lugar em $F/K\}$.
	\end{enumerate}
\end{defn}
Também  um lugar dá origem a uma função especial que definiremos a seguir:
\begin{defn}
	Uma valorização discreta\index{Valorização! Discreta} de $F/K$ é uma função $v:F\emm \Z \uniao \{ \infty \} $ com as seguintes propriedades:
    	\begin{enumerate}
        	\item $v(x)=\infty \sse x=0$;
        	\item $v(xy)=v(x)+v(y)$, para qualquer $x,y \in F$;
        	\item $v(x+y) \maiori min\{v(x),v(y)\}$, para qualquer $x,y, \in F$;
        	\item Existe $z\in F$ com $v(z)=1$;
        	\item $v(a)=0$, para todo $a \in K$.
	\end{enumerate}
\end{defn}

\begin{defn}
    	Para $P\in \pf$ associamos uma função de valorização, $v_P:F\emm \Z\uniao \{\infty\}$, da seguinte maneira: dado um elemento principal $t \in P$  temos que todo elemento $0\neq z$ em $F$ é escrito de maneira única na forma $z=t^nu$, com $n \in \Z$ e 
	 $u\in \av_P^*$. Assim fazemos $v_P(z) = n$ e $v_P(0)=\infty$.
\end{defn}
Um primeiro resultado importante a respeito de $v_P$ é dado no seguinte teorema:
\begin{thm}
    	Seja $F/K$ um corpo de funções então:
    	\begin{enumerate}
        	\item[a.] Para qualquer $P\in \pf$, a função  de valorização $v_P$ definida como acima é uma valorização discreta que satisfaz:
		\begin{enumerate}
			\item[a.1.] $ \av_P = \{z \in F ; v_P(z) \maiori 0\};$
        		\item[a.2.] $ \av_P^* = \{ z \in F ; v_P(z) = 0 \};$
        		\item[a.3.] $ P = \{z \in F ; v_P(z) > 0\};$
        		\item[a.4.] Um elemento $x\in F$ é principal de $P$ se, e somente se, $v_P(x)=1$.
		\end{enumerate}
        	\item[b.] Qualquer anel de valorização $\av$ de $F/K$ é um subanel maximal de $F$.
    	\end{enumerate}
\end{thm}

Definiremos agora alguns elementos da teoria que serão muito utilizados em nossa explanação.
\begin{defn}\label{def:corpo_residuo}
    	Sejam $P \in \pf$ e  $F_P=\av_P /P$  o corpo de resíduos de $P$. A aplicação $x \Emm x(P)$  de $F$ em $F_P\uniao \{\infty\}$, onde $x(P)=x+P$ se $x\in \av_P$ e $x(P)=\infty$ se $x\notin \av_P$, é chamada de aplicação de resíduos\index{Aplicação! de Resíduos} com respeito a $P$.
\end{defn}
Pode-se provar que (veja em \cite{stich}) se $P\in \pf$ então $K\contido F_P$ e $F_P/K$ é uma extensão finita e assim definimos:
\begin{defn}    
	Se $P\in \pf$ definimos o grau de $P$ como sendo:
	$$gr(P) = [F_P: K].$$
\end{defn}
Os conceitos de zeros e pólos são de importância fundamental para a descrição dos códigos algébricos geométricos. 
\begin{defn}
    	Seja $z \in F$ e $P\in \pf$. Dizemos que $P$ é um zero\index{Zero} de $z$ se, e somente se, $v_P(z) >0$, e  $P$ é um pólo\index{Pólo} de $z$ se, e somente se, $v_P(z) <0$. Se $v_P(z) = m>0$ falamos que $P$ é um zero de ordem $m$ e se $v_P(z)=-m<0$ dizemos que é um pólo de ordem $m$.
\end{defn}

O próximo resultado, chamado de teorema da aproximação fraca, vai ser útil na descrição de um dos códigos que iremos definir.

\begin{thm}[Aproximação Fraca]\label{teo:fraca}
	Sejam $\,\,F/K\,\,$ um corpo $\,\,$ de $\,\,$ funções, \mbox{$P_1,\ldots, P_n \in \pf$} lugares de $F/K$, dois a dois distintos, $x_1,\ldots , x_n \in F$ e $r_1,\ldots , r_n \in \Z$. Então, existe $x\in F$ tal que:
	$$v_{P_i}(x-x_i)=r_i, \textrm{ para } i=1,\ldots , n.$$
\end{thm}

\begin{cor}
	O corpo de funções possui infinitos lugares.
\end{cor}
\section{O Corpo de Funções Racionais} \index{Corpo! de funções racionais}

Como definimos anteriormente, o corpo de funções racionais  é o próprio corpo de frações $K(x)/K$. Já tínhamos definido o anel de valorização
$$\av_{p(x)} = \left\lbrace  \dfrac{f(x)}{g(x)} ; f(x),g(x) \in K[x], \,\, p(x)\ndivide g(x) \right\rbrace, $$
para um polinômio $p(x)\in K[x]$ irredutível. Observamos agora que para esse anel, o ideal
$$P_{p(x)} = \left\lbrace  \dfrac{f(x)}{g(x)} ; f(x),g(x) \in K[x], \,\,p(x)|f(x), p(x)\ndivide g(x) \right\rbrace ,$$
é seu único ideal maximal. Consideremos agora um outro anel de valorização
$$\av_{\infty} = \left\lbrace  \dfrac{f(x)}{g(x)} ; f(x),g(x) \in K[x], \,\, gr(f(x))\menori gr(g(x)) \right\rbrace ,$$
que tem como ideal maximal
$$P_{\infty} = \left\lbrace  \dfrac{f(x)}{g(x)} ; f(x),g(x) \in K[x], \,\, gr(f(x))< gr(g(x)) \right\rbrace .$$

Na realidade estes são os únicos anéis de valorizações de $K(x)$ e tal fato é expresso no resultado:
\begin{thm}
	Os únicos lugares de $K(x)/K$ são os da forma $P_{p(x)}$  e $P_{\infty}$ como acima definidos.
\end{thm}

\section{Divisores}
Nessa seção estaremos preocupados com a definição dos divisores de um corpo de funções, os quais serão fundamentais para a construção de alguns dos códigos que trataremos neste trabalho.
\begin{defn}
	O grupo aditivo abeliano livre que tem como base livre os lugares de $F/K$ é denotado por $\df$ e é chamado de grupo de divisores \index{Divisores} de $F/K$. Os elementos de $\df$ são chamados de divisores e são da forma
    	$$D = \sum_{P\in \pf} n_P P, \textrm{ com } n_P\in \Z\textrm{, único e quase sempre nulo}.$$
\end{defn}
Chamamos de {\bf suporte de um divisor} \index{Suporte! do divisor} $D$ o conjunto dos lugares $P$ tais que \mbox{$n_P \neq 0$} e denotamos por $supp(D)$. Dados dois divisores $D=\sum n_PP$ e $D'=\sum n'_PP$  a soma deles é dada por:
$$ D + D' = \sum_{P\in \pf} (n_P + n'_P)P.$$
E mais, para cada $P\in \pf$, definimos $v_P(D)= n_P$, assim também podemos definir uma ordem parcial de $\df$ por
$$D_1 \menori D_2 \sse v_P(D_1) \menori v_P(D_2), \,\, \ptodo P \in \pf.$$

Chamamos de grau de um divisor \index{Grau! do divisor} $D$ ao número $gr(D) = \sum_{P\in \pf} v_P(D) gr(P).$

Definiremos agora os divisores principais, os quais trazem importantes conseqüências para a teoria. Pode-se provar, veja em \cite{stich}, que se $x\in F$ então o número de zeros e pólos de $x$ é finito e esse fato nos leva a seguinte definição:

\begin{defn}
	Seja $0\neq x \in F$ e denotemos por $Z$ e $N$ o conjunto de zeros e pólos de $x$ em $\pf$, respectivamente, assim definimos:
    	\begin{enumerate}
        	\item $(x)_0=\sum_{P\in Z} v_P(x)P$, o divisor de zeros de $x$;
        	\item $(x)_{\infty} = \sum_{P\in N}(-v_P(x))P$, o divisor de pólos de $x$;
        	\item $(x) = (x)_0 -(x)_{\infty}$, o divisor principal de $x$.
    	\end{enumerate}
\end{defn}

\begin{defn}
    	Definimos o {\bf grupo de divisores principais} \index{Grupo! de divisores principais} de $F/K$ como sendo $\Pf = \{ (x) ;  0\neq x \in F\}$ (este é um subgrupo de $\df$ desde que para $0\neq x,y \in F, \, (xy)=(x)+(y)$). O quociente $\cf = \df/\Pf$ é chamado grupo de classe de divisores e dizemos que $D$ é equivalente a $D'$, denotando por $D\sim D'$, se $[D]=[D']$, isto é, se\index{Divisores! equivalentes} $D'=(x)+D$ para algum $x\in \F$.
\end{defn}
Agora vamos definir um espaço vetorial associado a um divisor, o qual nos dará a definição de  dimensão desse divisor e também de um invariante muito importante para a teoria.
\begin{defn}
    	Para um divisor $A\in \df$ definimos o conjunto
    	$$\LL(A) = \{ x \in F ; (x) \maiori -A\} \cup \{0\}.$$
\end{defn}
\begin{lem}
    	Seja $A \in \df$ então temos
    	\begin{enumerate}
        	\item[a.] $\LL(A)$ é um $K-$espaço vetorial, de dimensão finita;
        	\item[b.] Se $A' \sim A$ então $\LL(A') \iso \LL(A)$.
    	\end{enumerate}
\end{lem}

\begin{defn}
    	Para $A\in \df$, definimos a dimensão do divisor \index{Dimensão! do divisor} $A$ como sendo $dim(A) = dim_K(\LL(A))$.
\end{defn}

\begin{thm}
	Qualquer divisor principal tem grau zero. Mais precisamente dado $x\in F \setminus K$ e $(x)_0,\, (x)_{\infty} $ denotando os divisores de zeros e pólos, respectivamente, do divisor de $x$, então:
	$$ gr((x)_0) = gr((x)_{\infty}) = [F:K(x)].$$
\end{thm}
Como conseqüência imediata deste teorema temos: 
\begin{cor}
	\label{corolario_grau_1}
	\begin{enumerate}
		\item Seja $A,\,A'$ divisores tais que $A\sim A'$. Assim temos que $dim(A) = dim(A')$ e que $gr(A) = gr(A')$;
		\item Se $gr(A) < 0$ então $dim(A) = 0$;
		\item Para um divisor $A$ de grau zero , temos que são equivalentes:
		\begin{enumerate}
			\item $A$ é um divisor principal;
			\item $dim(A)>0;$
			\item $dim(A)=1.$
		\end{enumerate}
	\end{enumerate}
\end{cor}
A partir da proposição que enuciaremos a seguir definimos o gênero de um corpo de funções algébricas que é um invariante que depende apenas do corpo.

\begin{prop}
	Existe uma constante $\gamma\in \Z$ tal que, para todo divisor $A\in \df$ temos que:
	$$ gr(A) - dim(A) \menori \gamma.$$
\end{prop}

\begin{defn}
    	Definimos o gênero\index{Gênero} do corpo de funções $F/K$ como sendo
    	$$g=max\{gr(A) - dim(A) +1 ; A\in \df \}.$$
\end{defn}
O  primeiro resultado envolvendo o gênero de um corpo de funções é o seguinte:

\begin{thm}[Teorema de Riemann]\index{Teorema! de Riemann}
    	Seja $F/K $ uma corpo de funções algé\-bricas com gênero $g$ então:
    	\begin{enumerate}
        	\item Para qualquer divisor $A \in \df$, $dim(A) \maiori gr(A)+1-g$;
        	\item Existe um inteiro $c$, dependendo de $F/K$ tal que $dim(A) = gr(A) +1 - g$ sempre que $gr(A) \maiori c$.
    	\end{enumerate}
\end{thm}
\section{O Teorema de Riemann-Roch} 
Aqui $F/K$ denota um corpo de funções algébricas com gênero $g$.

\begin{defn}
    	Para $A \in \df$ definimos o índice de especialidade \index{Índice! de especialidade} de $A$ como sendo:
    	$$ i(A) = dim(A) - gr(A) +g-1.$$
\end{defn}

Como se vê o Teorema de Riemann garante que o índice de especialidade de um divisor $A$ é inteiro não negativo e na verdade nós vamos apresentar $i(A)$ como a dimensão de certos espaços vetoriais. Assim para isto começamos com as definições abaixo:
\begin{defn}
    	Um adele \index{Adele} de $F/K$ é uma função
    	$$ \funcao{\alpha}{\pf}{F}{P}{\alpha_P} $$
    	tal que $\alpha_P \in \av_P$ para quase todos $P \in \pf$. Podemos  ver um adele como um elemento  do produto direto $\prod_{P\in \pf}F$ e usamos a notação $\alpha = (\alpha_P)_{P\in\pf}$, e para encurtar, $\alpha = (\alpha_P)$.
\end{defn}
Chamamos de espaço de adeles \index{Espaço! de adeles} de $F/K$ ao
conjunto:
$$\af = \{\alpha ;\,\, \alpha \textrm{ é adele de } F/K \}.$$

\begin{defn}
	\begin{enumerate}
		\item  	Definimos como adele principal de um elemento $x \in F$ como sendo o adele cujas as componentes são todas iguais a $x$ ou seja a seqüência constante $(x)$.
		\item Dado $P\in \pf$ e $x\in F$, o adele cuja as componentes são nulas a menos da componente $P$, a qual é $x$, será denotado por $\iota_P(x)$.
	\end{enumerate}
\end{defn}
\begin{defn}
	A um adele $\alpha$, associamos uma função de valorização discreta dada por: $v_P(\alpha):=v_P(\alpha_P)$.
\end{defn}

\begin{defn}
    	Para um divisor $A \in \df$ definimos o conjunto:
    	$$ \af(A) = \{\alpha \in \af ;\,\, v_P(\alpha) \maiori -v_P(A), \,\, \forall P \in \pf\}.   $$
\end{defn}
Facilmente pode se ver que o conjunto acima definido é um $K-$subespaço vetorial de $\af$.

No próximo teorema explicitamos um espaço vetorial cuja a dimensão é o índice de especialidade do divisor $A$.
\begin{thm}
    	Dado um divisor $A$, o seu índice de especialidade é dado por:
    	$$i(A) = dim_{K}(\af/(\af(A)+F)).$$
\end{thm}
E como conseqüência imediata temos:
\begin{cor}
    	$$ g = dim_{K}(\af /(\af(0)+F)).$$
\end{cor}
Agora veremos o conceito de diferenciais de Weil, o qual nos dará mais informações sobre o indice de especialidade de um divisor.
\begin{defn}
    	Uma aplicação $K-$linear, $\omega:\af \emm K$, que se anula em $\af(A)+F$ para algum divisor $A\in \df$ é dita diferencial de Weil de $F/K$ \index{Diferencial! de Weil}. Denotamos por $\Omega_F $ ao    conjunto dos diferenciais de Weil de $F/K$  e por $\Omega_F(A)$   ao conjunto de diferenciais de Weil de $F/K$ que se anulam em    $\af(A) + F$.
\end{defn}

O conceito de divisor canônico é de fundamental importância para o Teorema de Riemann-Roch e ele é definido a partir da próxima definição e do lema seguinte.

\begin{defn}
    	Para um diferencial de Weil $\omega \neq 0$ definimos o conjunto   de divisores:
    	$$M(\omega) = \{ A \in \df ;\,\, \omega \textrm{ se anula em }    \af(A) +F\}.$$
\end{defn}
\begin{lem}
    	Seja $\omega\in \Omega_F,\,\, \omega \neq 0$, então existe um único divisor $W \in  M(\omega)$ tal que $A \menori W$ para todo $A \in M(\omega)$.
\end{lem}

\begin{defn}[Divisor Canônico]
    	\begin{enumerate}
        	\item O divisor $(\omega)$ de um diferencial de Weil $\omega
        	\neq 0$ é o divisor de $F/K$ unicamente determinado por:
        	\begin{enumerate}
            		\item $\omega$ se anula em $\af((\omega))+F$;
            		\item Se $\omega$ se anula em $\af(A) + F$ então $A \menori (\omega).$
        \end{enumerate}
        \item Para $0\neq \omega \in \Omega_F$ e $P \in \pf$        definimos $v_P(\omega)=v_P((\omega))$;

        \item Um lugar $P$ é dito zero (resp. pólo) de $\omega$ se        $v_P(\omega) > 0$ (resp. $v_P(\omega)<0$). $\omega$ é        chamado de regular em $P$ se $v_P(\omega) \maiori 0$, e        simplesmente de regular se for regular em todos os lugares   em $\pf$;
        \item Um divisor $W$ é dito divisor canônico \index{Divisor!  canônico} de $F/K$ se $W=(\omega)$ para algum $\omega \in        \Omega_F$.
    \end{enumerate}
\end{defn}
\begin{prop}
	\begin{enumerate}
		\item Para $0\neq x \in F$ e $0\neq \omega \in \Omega_F$ temos que $(x\omega)=(x)+(\omega)$;
		
		\item Quaisquer dois divisores canônicos são equivalentes.
	\end{enumerate}
\end{prop}
Esse próximo teorema vai nos permitir calcular o indice de especialidade de um divisor (mais uma vez como dimensão de um espaço vetorial).
\begin{thm}
	Seja $A$ um divisor arbitrário e $W=(\omega)$ um divisor canônico de $F/K$. Então a aplicação
	$$\funcao{\mu}{\LL(W-A)}{\Omega_F(A)}{x}{x\omega},$$
	é um isomorfismo de $K-$espaços Vetoriais, e mais $i(A)=dim(W-A)$.
\end{thm}

Em fim chegamos ao principal teorema da teoria até aqui apresentada.
\begin{cor}[Riemann-Roch]\index{Teorema! Riemann-Roch}\label{teo:rie_roc}
	Seja $W$ um divisor canônico de $F/K$. Assim para qualquer $A \in \df$ temos que:
	$$dim(A) = gr(A) +1 -g + dim(W-A).$$
\end{cor}
\begin{thm}\label{teorema_dimensao_1}
	Sendo $A$ um divisor de $F/K$ de grau maior ou igual  que $2g-1$, temos:
	$$ dim(A) = gr(A) +1-g.$$
\end{thm}

A próxima definição nos será bastante útil para o entendimento de um de nossos códigos.

\begin{defn}[Componente Local]\index{Componente! Local}
	Seja um diferencial de Weil $\omega \in \Omega_F$. Definimos a componente local de $\omega$ como sendo a aplicação $K-linear$
	$$\funcao{\omega_P}{F}{K}{x}{\omega(\iota_P(x))}.$$
\end{defn}

Sobre as componentes locais temos os seguintes resultados:
\begin{prop}\label{prop:weil_soma}
	Sejam um diferencial de Weil $\omega \in \Omega_F$ e um adele $\alpha=(\alpha_P)\in \af$. Então $\omega_P(\alpha_P)\neq 0$  no máximo em finitos lugares $P$, e mais,
	$$\omega(\alpha)=\sum_{P\in \pf} \omega_P(\alpha_P).$$
	Em particular,
	$$\sum_{P\in \pf} \omega_P(1)=0.$$
\end{prop}
 
O próximo resultado nos mostra que um diferencial de Weil é unicamente determinado em relação aos seus componentes locais.
 
\begin{prop}\label{prop:weil1}
	\begin{enumerate}
		\item Seja $\omega\neq 0$ um diferencial de Weil de $F/K$, $P\in \pf$ e $W=(\omega)$. Então:
		$$v_P(W) = max\{r\in \Z ; \, \omega_P(x)=0 \textrm{ para todo } x\in F \textrm{ com } v_P(x)\maiori -r\}.$$
		
		\item Se $\omega, \, \omega' \in \Omega_F$ e $\omega_P = \omega_p'$ para algum $P\in \pf$, então $\omega=\omega'$.
	\end{enumerate}

\end{prop}
\section{Curvas Algébricas}
Nessa seção apresentaremos as curvas algébricas e sua ligação com os corpos de funções algébricas.
\begin{defn}
	Seja $I$ um ideal em $\F[x_1,\cdots,x_n]$ (anel de polinômios em $n$ indeterminadas com coeficientes sobre um corpo  $\F$). Definimos o conjunto algébrico\index{Conjunto! algébrico} obtido a partir de $I$ como sendo:
	$$V(I) = \{a=(a_1,\cdots,a_n) \in \F^n ; \,\, f(a) = 0 ,\, \ptodo f \in I\}.$$
\end{defn}
Dizemos que um conjunto algébrico $B$ é irredutível se não pode ser escrito como união de dois outros conjuntos algébricos próprios. Temos que $V(I)$ é irredutível quando o radical de $I$ é ideal primo. Tal resultado nos leva as duas  definições abaixo:

\begin{defn}
	\begin{enumerate}
		\item Dado um ideal primo $I\,\contido \,\F[x_1,\, \cdots\,,x_n]$ o conjunto \mbox{$\X\, =\, V(I)$} é dito variedade afim\index{Variedade! afim}.
		\item O anel $\F[x_1,\cdots,x_n]/I = \F[\X]$ é dito anel de coordenadas de $\X$.\index{Anel! de coordenadas}
	\end{enumerate}
\end{defn}

\begin{defn}
	Dada uma variedade algébrica $\X$ definimos  o corpo de funções racionais de $\X$ \index{Corpo! de funções racionais de $\X$} como sendo o corpo de frações do anel de coordenadas $\F[\X]$ que é denotado por $\F(\X)$.
\end{defn}
Um resultado clássico da álgebra comutativa (Teorema de normalização de Noether) nos diz que podemos dar a dimensão da variedade algébrica da seguinte forma:
\begin{defn}
	Definimos a dimensão da variedade\index{Dimensão! da variedade} $\X$ como sendo o grau de transcendência de $\F(\X)/\F$.
\end{defn}
Para um ponto $P\in \X$, o conjunto 
$$\av_P(\X) = \left\lbrace  f \in \F(\X) ; f = \dfrac{g}{h} , \,\, g,h \in \F[\X] \textrm{ e } h(P)\neq 0  \right\rbrace ,$$
é um anel local que tem como corpo de  frações o próprio $\F(\X)$ e mais o seu ideal maximal é dado por:
$$M_P(\X)=\left\lbrace  f \in \F(\X) ; f = \dfrac{g}{h} , \,\, g,h \in \F[\X], g(P)=0 \textrm{ e } h(P)\neq 0    \right\rbrace .$$

\subsection{Variedades Projetivas}

\begin{defn}
	Dado um corpo $\F$, no conjunto $\F^{n+1} \setminus \{0\}$ definimos a  relação de equivalência $\sim$ definida por:  dados dois vetores $v=(v_0,v_1,\cdots,v_n)$ e $w=(w_0,w_1,\cdots,w_n) \in \F^{n+1}\setminus \{0\}$ 
	eles são equivalentes se forem linearmente dependentes sobre $\F$, ou seja, $v \sim w$ se, e somente se,  existe $\lambda \in \F$ tal que $v=\lambda w$.
\end{defn}
O conjunto quociente $(\F^{n+1} \setminus \{0\})/\sim$ das classes de equivalência segundo a relação $\sim$,   é chamado {\bf espaço projetivo\index{Espaço! projetivo} de dimensão $n$} e denotado por $\proj{n}(\F)$ e seus elementos são denotados por $(a_0:a_1:\cdots :a_n)$.

\begin{defn}
	Dizemos que um polinômio $F \in \F[x_1,\cdots,x_n]$ é homogêneo\index{Polinômio! Homogêneo} se esse for soma de monômios de mesmo grau. Um ideal gerado por polinômios homogêneos é chamado de {\bf ideal homogêneo}\index{Ideal! homogêneo}.
\end{defn}

Observe que dados um ponto $P=(a_0:a_1:\cdots:a_n)=(b_0:b_1:\cdots:b_n)\in \proj{n}(\F)$ e um polinômio homogêneo $F\in \F[x_0,\cdots,x_n]$ podemos definir  $F(P)=0$ se $F(a_0,a_1,\cdots,a_n)=0$, já que $F(a_0,a_1,\cdots,a_n)=0$ se, e somente se, $F(b_0,b_1,\cdots,b_n)=0$. Assim, podemos definir o que seja uma variedade algébrica projetiva.

\begin{defn}
	Um subconjunto $\X\contido \proj{n}(\F)$ é dito uma {\bf variedade algébrica projetiva}\index{Variedade! algébrica projetiva} se for o conjunto de zeros de um ideal homogêneo $I\contido \F[x_0,x_1,\cdots,x_n]$, ou seja:
	$$ \X=V(I) = \{ P \in \proj{n}(\F); \,\, F(P)=0,\, \ptodo F \in I\}.$$
\end{defn}
Uma variedade algébrica projetiva $\X =V(I)$ é irredutível se, e somente se, o ideal $I$ for um ideal homogêneo e o seu radical for primo.

O anel de coordenadas $\F_h[\X]=\F[x_0,\cdots,x_n]/I$ é dito anel de coordenadas homogêneas e os  elementos que são do formato $f=F+I$ com $F \in \F[x_0,\cdots,x_n]$ e $F$ homogenea,são chamados de {\bf forma de grau $d$} onde $d=gr(F)$.

Faremos agora a associação entre as curvas algébricas e os corpos de funções.
\begin{defn}
	Se $\X$ é uma variedade algébrica projetiva definimos o corpo de funções de $\X$ como sendo:
	$$\F(\X) = \left\lbrace   \frac{g}{h} ; \,\, g,h \in \F_h[\X] \textrm{ formas de mesmo grau e } h\neq 0 \right\rbrace .$$
\end{defn}
A dimensão da variedade $\X$ é dada pelo grau de transcendência de $\F(\X)/\F$.
\begin{defn}
	Dado um ponto $P\in \X$ e $f=\frac{g}{l} \in \F(\X)$ com $g,l \in \F_h[\X]$, dizemos que $f$ está definida em $P$ se $l(P)\neq 0$, $f(P)$ é dito valor de $f$ em $P$.
\end{defn}
O anel $\av_P(\X)=\{f\in\F(\X); \,\, f \textrm{ é definida em } P \} \contido \F(\X)$ é um anel local com ideal maximal $M_P(\X)=\{f\in \av_P(\X) ; \,\, f(P)=0\}$.
\begin{defn}
	Seja $F\in \F[x_1,\cdots,x_n]$ com grau total de $F$ igual a $l$, definimos a homogeneização de $F$ como sendo:
	$$F^* = x_0^l F\left( \frac{x_1}{x_0},\cdots,\frac{x_n}{x_0}\right) .$$
\end{defn}
\begin{defn}
	 Uma variedade projetiva (afim) $\X$ de dimensão 1 é chamada de curva algébrica irredutivel projetiva (afim). Desse modo o corpo de funções racionais em $\X$, $\F(\X)$, é um corpo de funções algébricas de uma variável como na definição \ref{corpo_funcao}. Mais ainda, dizemos que a curva algébrica projetiva (afim) é plana se $\X \contido \proj{2}(\F)$ (ou $\F^2$).
\end{defn}

\subsection{Curvas Não Singulares}

Levando-se em conta que os códigos geométricos de Goppa são gerados pelas  curvas planas  não singulares, a seguir apresentaremos a definição de curvas não singulares.

\begin{defn}
	Seja $\X$ uma curva algébrica definida pelo polinômio $F\in \F(x,y)$ e $P$ um ponto em $\X$. Dizemos que o ponto $P$ é um ponto não singular \index{Ponto! não singular} se pelo menos uma das derivadas parciais de $F$ aplicadas nesse ponto é não nula, ou seja, $F_x(p) \neq 0$ ou $F_y(P)\neq 0$. Se todos os pontos da curva forem não singulares dizemos apenas que a curva é não singular (ou regular).\index{Curva! não singular}
\end{defn}

Observamos que aqui que a derivada parcial de um polinômio é a sua derivada formal.


\chapter{Códigos Algébricos Geométricos}
 Neste capítulo estamos interessados em construir os códigos algébricos geomé\-tricos, conhecidos como Códigos de Goppa Geométricos, e então extrair  uma cota para a sua distância mínima e calcular sua dimensão. Aqui $\F_q$ denotará um corpo com $q$ elementos.
\section{Códigos Algébricos Geométricos}
\begin{defn}
	Seja $\F_q=\{\alpha_0,\, \cdots,\, \alpha_{q-1}\} $ e considere o conjunto, $\LL_k \contido \F_q[x]$, dos polinômios com grau menor que $k$ e $k\menori q$, definimos um código de Reed-Solomon de tamanho $n=q$ como sendo:
	$$C_k = \{ c(f)=(f(\alpha_0),\cdots,f(\alpha_{q-1}));\, f\in \LL_k \}.$$
\end{defn}
\begin{prop}
	Como acima definido, $C_k$ é um código MDS (Maximum Distance Separable), ou seja, tem distância mínima $d=n-k+1$.
\end{prop}
\begin{proof}
	Seja $f\in \LL_k$, temos que $f$ tem grau no máximo $k-1$, assim $f$ admite no máximo $k-1$ raízes distintas em $\F_q$, desse modo o peso de $c(f)$ é no mínimo $n-k+1$, ou seja, $d\maiori n-k+1$ no entanto temos que o polinômio $g(x) = (x-\alpha_{i_1}) \cdots (x-\alpha_{i_{k-1}}) \in \LL_k$  tem grau $k-1$ e tem exatamente $k-1$ raízes distintas em $\F_q$, desse modo $d(c(g))=n-k+1$.
\end{proof}

A seguir utilizaremos as seguintes notações:
\begin{enumerate}
	\item[$\bigstar$]  $F/\F_q$, um corpo de funções algébricas de gênero $g$;
	\item[$\bigstar$]  $P_1,P_2,\cdots, P_n$, lugares de grau 1 dois a dois distintos em $F/\F_q$;
	\item[$\bigstar$]  $D=\sum_{i=1}^n P_i$, um divisor em $F/\F_q$;
	\item[$\bigstar$]  $G$, um divisor em $F/\F_q$ tal que $supp(G) \inter supp(D) = \emptyset$.
\end{enumerate}
A partir das notações dadas acima e da definição \ref{def:corpo_residuo} estamos em condição de definir o que vem a ser um código gemétrico de Goppa.
\begin{defn}[Códigos Geométricos de Goppa]\label{cod:goppa}
	Definimos o código algébrico  geométrico (código geométrico de Goppa) associado aos divisores $D$ e $G$ como sendo:
	$$C(D,G) = \{ c(x) = (x(P_1),x(P_2),\cdots,x(P_n)) ; x \in \LL(G)\}.$$
\end{defn}

\begin{thm}\label{teo:goppaparametro}
	O código de Goppa $C(D,G)$ é um $[n,k,d]-$código tal que:
	$$k=dim(G)-dim(G-D) \textrm{ e } d\maiori n-gr(G).$$
\end{thm}

\begin{proof}
	Seja a aplicação de avaliação
	$$\funcao{ev_D}{\LL(G)}{C(D,G)}{x}{c(x)=((x(P_1),\cdots,x(P_n))}.$$
	Sabemos que $ev_D$ é sobrejetiva logo $\dfrac{\LL(G)}{Ker(ev_D)}\iso C(D,G)$. Queremos agora mostrar que $Ker(ev_D) = \LL(G-D)$. Seja $x\in Ker(ev_D)$, temos que para todo $Q\in \pf$ $v_Q(x) \maiori -v_Q(G)$ e $v_{P_i}(x) > 0 $ para $i=1, \cdots, n$ e mais, como $supp(G) \inter supp(D) = \emptyset$, temos que, $v_Q(x) \maiori -v_Q(G) + v_Q(D)$ assim $Ker(ev_D) \contido \LL(G-D)$.
	
	Tomemos agora $x \in \LL(G-D)$, temos que $v_{P_i}(x) \maiori -v_{P_i}(G) + v_{P_i}(D) = 1$ para $i=1,\cdots, n$, logo $x \in \LL(G)$ e $v_{P_i}(x)>0$ ou seja $x(P_i)=0$ em $F_{P_i}$ o que implica que $\LL(G-D) \contido Ker(ev_D)$ e pelo teorema dos isomorfismos (álgebra linear) concluímos que $k=dim(G)-dim(G-D).$
	
	Agora acharemos uma cota para a distância mínima do código. Assumiremos que $C(D,G) \neq \{0\}$, pois caso contrário não teriamos uma distância não nula.
	
	Seja $x\in C(D,G)$ tal que $w(x)=d$, assim temos que existe $y\in \LL(G)$ tal que $w(ev_D(y))=d$, mas $w(ev_D(y)) = d = \#\{ i ; y(P_i)\neq 0\}$. Suponha agora que $y(P_1)=y(P_2)=\cdots = y(P_{n-d})=0$, reordenando os $P_i$'s caso necessário, logo
	$$ 0\neq y \in \LL \left( G - \sum_{i=1}^{n-d}P_i \right) \implica dim\left( G - \sum_{i=1}^{n-d}P_i \right) \neq 0 \implica$$
	$$ 0 \menori gr \left( G - \sum_{i=1}^{n-d}P_i \right) = gr(G) -n +d \implica$$
	$$ d \maiori n-gr(G).$$
\end{proof}
\begin{cor}
	Suponha que $gr(G)<n$ então $ev_D:\LL(G) \emm C(D,G)$ é injetora e se $2g-2 < gr(G) <n$ teremos que $k=gr(G) +1-g$.
\end{cor}
\begin{proof}
	Temos que $gr(G-D)=gr(G) - n <0$, logo  (por \ref{corolario_grau_1}-2) $ \{0\} = \LL(G-D) = Ker(ev_D) $ logo $ev_D$ é injetora. Pelo teorema \ref{teorema_dimensao_1} temos que $k=gr(G)+1-g$.
\end{proof}

Esses resultados dão uma motivação para a seguinte definição:

\begin{defn}
	No código geométrico de Goppa $C(D,G)$ chamamos de distân\-cia designada ao inteiro $d^*=n-gr(G)$.
\end{defn}

Ao ver essa definição surge uma pergunta: Quando a distância designada é igual a distância mínima do código? Responderemos essa pergunta na próxima proposição.

\begin{prop}
	Seja $C(D,G)$ um código com distância designada $d^*$, suponha que $dim(G)>0$ e que $d^*>0$. então $d=d^*$ se, e somente se, existe um divisor $D'$ com $0\menori D' \menori D$, $gr(D')=gr(G)$ e $dim(G-D')>0$.
\end{prop}
\begin{proof}
	Suponha que $d=d^*$. Seja $0\neq x \in \LL(G)$ tal que $w(ev_D(x))=d$. Assim, $ev_D(x)=(x(P_1),\ldots,x(P_n))$, reordenando os $P_i's$ caso necessário, temos $x(P_1)=\cdots=x(P_{gr(G)})=0.$ Seja $D'=\sum_{i=1}^{gr(G)} P_i$, logo $gr(D')=\sum_{i=1}^{gr(G)} v_{P_i}(D')=gr(G).$ Assim temos que $0\menori D' \menori D.$ Observe que $x \in Ker(ev_{D'})$ logo como na demonstração do teorema \ref{teo:goppaparametro}, tem-se $dim(G-D')>0.$
	
	Queremos agora demonstrar a segunda parte da proposição. Seja $D'\in \df$ tal que $0\menori D' \menori D$, $gr(D')=gr(G)$ e $dim(G-D')>0$, então temos que existe $y\in \LL(G-D')$, assim sendo, temos que o peso da palavra código $(y(P_1),\ldots,y(P_n))$ é no máximo $n-gr(G)=d^*$, contudo $d$ é a distância mínima no código e mais, como provado anteriormente $d^*\menori d$, logo $d=d^*$.
\end{proof}

Agora definiremos o código que originalmente foi introduzido por V. D. Goppa em 1981 no artigo ``Codes on Algebraic Curves''.

\begin{defn}
	Seja $G$ e $D=P_1 + \cdots + P_n$ divisores de $F/\F_q$ com os $P_i's$ dois a dois disjuntos e $supp(G) \inter supp(D) = \emptyset$. Definimos o código $C_{\Omega}(D,G)$ por:
	$$C_{\Omega}=\{(\omega_{P_1}(1),\ldots,\omega_{P_n}(1))\in \F_q^n;\,\, \omega\in \Omega(G-D)\}.$$
\end{defn}

O próximo resultado tem como objetivo caracterizar o código $C_{\Omega}(D,G)$. Mas antes disso apresentaremos um lema técnico.

\begin{lem}\label{lema:1}
	Sejam $F/\F_q$ um corpo de funções, $ P\in \pf$ com $gr(P)=1$ e $\omega \in \Omega_F$ um diferencial de Weil tal que $v_P(\omega)\maiori -1$. Então $\omega_P(1)=0$ se, e somente se, $v_P(\omega)\maiori 0 $.
\end{lem}
\begin{proof}
	$\volta)$ Primeiro vamos supor que $v_P(\omega)\maiori 0$, logo pela proposição \ref{prop:weil1}, temos que $\omega_P(x)=0$ para todo $x\in F$ com $v_P(x)\maiori 0$, temos que $1\in \F_q\contido F$ logo $v_P(1)=0$ assim $\omega_P(1)=0$.
	
	$\implica)$ Suponha agora que $\omega_P(1)=0$, assim temos que $\omega_P(a)=a\omega_P(1)=0$ para todo $a\in \F_q$. Pelo teorema \ref{teo:fraca}, existe $x\in F$ tal que $v_P(x)\maiori 0$, ou seja $x\in \av_P$. Como $gr(P)=1$ temos que $\av_P / P = \F_q$, assim existem $y\in P$ e $a \in \F_q$ tais que $x-y=a$ e mais, $v_P(y)\maiori 1$ e $v_P(a)=0$. Por hipótese temos que $v_P(\omega)\maiori-1$ e como $v_P(y)\maiori 1 $, pela proposição \ref{prop:weil1}, temos que $\omega_P(y)=0$. Assim temos que $\omega_P(x)=\omega_P(a+y) = \omega_P(a)+\omega_P(y) = 0$, ou seja, $v_P(\omega)\maiori 0$ (novamente pela proposição \ref{prop:weil1}).
\end{proof}

\begin{obs}\label{obs:1}
	A proposição \ref{prop:weil1} nos garante que $v_P(\omega)\maiori r$ se, e somente se, $\omega(x)=0$ para todo $x\in F$ com $v_P(x) \maiori -r$.
\end{obs}

Agora enuciaremos o teorema que caracteriza o código $C_{\Omega}(D,G)$.

\begin{thm}\label{teo:para2}
	O código $C_{\Omega}(D,G)$ é um $[n,k',d']-$código com parâmetros:
	$$k'=i(G-D)-i(G) \textrm{ e } d'\maiori gr(G) - 2g +2.$$
	E mais, se $gr(G) > 2g-2$, temos que $k'=i(G-D)\maiori n+g-1-gr(G)$ e se $2g-2<gr(G)<n$ então $k'=n+g-1-gr(G).$
\end{thm}

\begin{proof}
	Seja $\phi$ a seguinte aplicação:
	$$\funcao{\phi}{\Omega_F(G-D)}{C_{\Omega}(D,G)}{\omega}{(\omega_{P_1}(1),\ldots,\omega_{P_n}(1))},$$
	obviamente $\phi$ é sobrejetiva, logo $C_{\Omega}(D,G)$ é isomorfo a $\Omega_F(G-D)/Ker(\phi)$.
	
	Seja $\omega \in Ker(\phi)$, temos que $\omega_{P_i}(1)=0$ para $i=1\ldots n$, assim pelo lema \ref{lema:1}, $v_{P_i}(\omega)\maiori 0$. Como $\omega \in \Omega_F(G-D)$, temos que $(\omega)\maiori G-D$. Observe que se $P\in \{P_1,\ldots,P_n\}$ temos que $v_P(G)=0$ e $v_P(-d)<0$ e mais, para $P\notin \{P_1,\ldots,P_n\}$ temos que $v_P(-D)=0$ assim $\omega \in \Omega_F(G)$.
	
	Seja agora $\omega \in \Omega_F(G)$, novamente pelo lema \ref{lema:1} temos que $v_{P_i}(\omega) \maiori v_{P_i}(G) = 0$, logo $\omega \in Ker(\phi).$
	
	Assim pelo isomorfismo visto acima temos que
	$$k'=dim_{\F_q}(\Omega_F(G-D))-dim_{\F_q}(\Omega_F(G)) = i(G-D) - i(G).$$

	Seja $\phi(\omega)\in C_{\Omega}(D,G)$, uma palavra com peso $m >0$,  sem perda de generalidade podemos supor que $v_{P_1}=\cdots=v_{P_{n-m}}$, desse modo temos que $\omega\in \Omega_F(G-(D-\sum_{i=1}^{n-m} P_i))$.
	
	Observe que $i(G-(D - \sum_{i=1}^{n-m}))=dim_{\F_q}(G-(D - \sum_{i=1}^{n-m}))$ e como $\Omega(G-(D - \sum_{i=1}^{n-m}))\neq \{0\}$ temos que $i(G-(D - \sum_{i=1}^{n-m}))>0$, desse modo $dim(G-(D - \sum_{i=1}^{n-m})) >gr(G-(D - \sum_{i=1}^{n-m})) - g +1$ o que é a contra positíva do teorema \ref{teorema_dimensao_1}, assim 
	$$2g-2 \maiori gr\left( G-\left(D - \sum_{i=1}^{n-m}\right) \right) = gr\left(G\right) - gr\left(D - \sum_{i=1}^{n-m}\right) = gr\left(G\right)-m$$
	concluimos então que $m\maiori gr(G)-2g+2$, ou seja, $d'\maiori gr(G)-2g+2$.
	
	Assuma agora que $gr(G)>2g-2$. Pelo teorema \ref{teorema_dimensao_1}, temos que $i(G)=0$, assim $k'=i(G-D)=dim(G-D)-gr(G-D)+g-1$, como $dim(G-D)\maiori 0$ temos que $k'\maiori n+g-1-gr(G)$.
	
	Agora se $gr(G)<n$ temos que $gr(G-D) = gr(G)-n < 0$ logo pelo corolário \ref{corolario_grau_1} $dim(G-D)=0$ assim $k'=n+g-1-gr(G)$.
\end{proof}

O próximo teorema mostra a ligação entre o código geométrico de Goppa ($C(D,G)$) e o código que acabamos de definir $C_{\Omega}(D,G)$. Novamente antes do teorema apresentaremos mais um lema técnico.

\begin{lem}\label{lema:2}
	Sejam $P\in \pf$ tal que $gr(P)=1$, $\omega $ um diferencial de Weil com $v_P(\omega)\maiori -1$ e $x\in F$ com $v_P(x)\maiori 0$, então 
	$$ \omega_P(x) = x(P)\omega_P(1).$$
\end{lem}
\begin{proof}
	Como $gr(P)=1$ temos que $\av_P/P = \F_q$, o fato de $v_P(x)\maiori 0$ nos dá que $x\in \av_P$, logo existe $y\in P$ e $a \in \F_q$ tais que $x-y=a$, observe que $v_P(a)=0$ e $v_P(y)\maiori 1$. Como $v_P(y)\maiori 1$ temos que $\omega_P(y)=0$, assim $\omega_P(x)=\omega_P(a+y)=\omega_P(a)+\omega_P(y) = a\omega_P(1).$
	
	Observe agora que $a$ é a classe de resíduos de $x$ em relação a $P$, logo $\omega_P(x)=x(P)\omega_P(1).$
\end{proof}

Em fim o teorema que nos dá a relação entre os códigos.

\begin{thm}
	Os códigos $C(D,G)$ e $C_{\Omega}(D,G)$ são duais entre si, ou seja,
	$$C_{\Omega}(D,G)=C(G,D)^{\bot}.$$
\end{thm}
\begin{proof}
	Primeiro vamos mostrar que a dimensão dos códigos ($C(D,G) ^{\bot}$ e $C_{\Omega}(D,G)$) são iguais.
	
	Pelo teorema \ref{teo:para2} temos que $dim(C_{\Omega}(D,G)) = i(G-D)-i(G)$, agora pelo corolário \ref{teo:rie_roc} (Riemann-Rock) temos que $i(G-D)-i(G) = dim(G-D)+g-1-gr(G-D) -(dim(G)+g-1-gr(G))=dim(G-D) +gr(D) -dim(G)= n+dim(G-D)-dim(G)$, pelo teorema \ref{teo:goppaparametro} temos que $dim(C(D,G))=dim(G)-dim(G-D)$, assim $dim(C_{\Omega}(D,G)) = dim( C(D,G)^{\bot})$.
	
	Agora mostraremos que $C_{\Omega}(D,G) \contido C(D,G)^{\bot}$, o que vai nos garantir a igualdade desejada.
	
	Seja $\omega \in \Omega_F(G-D)$  e $x\in \LL(G)$. Identificando $x$ como um adele principal temos que $0=\omega(x)= \sum _{P\in\pf}\omega_P(x)$ (pela proposição \ref{prop:weil_soma}). Para $P\in \pf \setminus \{P_1,\ldots,P_n\}$ temos que $v_P(x)\maiori -v_P(\omega)$, então pela observação \ref{obs:1} temos que $v_P(x)=0$, assim $\sum_{P\in\pf}\omega_P(x) =\sum _{i=1}^n\omega_{P_i}(x)$. Agora pelo lema \ref{lema:2} temos que $\sum _{i=1}^n\omega_{P_i}(x) = \sum_{i=1}^n x(P_i)\omega_{P_i}(1) = \gera{(\omega_{P_1},\ldots,\omega_{P_n}),(x(P_1),\ldots,x(P_n))}$, concluindo então que $C_{\Omega}(D,D) \contido C(D,G)^{\bot}$.
\end{proof}

\section{Teorema de Bézout}
Nesta seção ficaremos a par do teorema de Bézout, o qual fala sobre o número de interseções entre curvas algébricas e também o ligaremos a teoria de códigos,  foco deste trabalho.

\begin{thm}[Teorema de Bézout]\index{Teorema! de Bézout}
	Sejam $\X$ e $\Y$ duas curvas algébricas planas irredutíveis de grau $l$ e $m$ respectivamente sobre um corpo algebricamente fechado $\F$ tais que elas não tenham uma componente em comum, então o número de pontos da interseção entre as curvas é exatamente $lm$ (contanto os pontos com suas  multiplicidades).
\end{thm}
\begin{proof}
	A demonstração desse resultado não está no objetivo desse texto, contudo ela se encontra em \cite{fulton}.
\end{proof}

\begin{prop}\label{teor:irre_fecho}
	Consideremos um polinômio $G\in \F_q[x,y]$ com grau total $m$ tal que sua forma homogênea $G^*$ define uma curva irredutível não singular $\X$, então $G$ é  irredutível em $\F[x,y]$, onde $\F$ é o fecho algébrico de $\F_q$.
\end{prop}
\begin{proof}
	Pela nossa definição $\X$ é uma curva gerada por um ideal primo $I\contido \F[x,y,z]$ e mais $I=\langle G^* \rangle$. Logo $G^*$ é irredutível em $\F[x,y,z]$. Agora sendo $G^*$ irredutível e supondo, por absurdo, que $G$ seja redutível temos que $G=fh$ com $gr(f)<gr(G)$ e $gr(h)<gr(G)$, logo temos que $G^* = (fh)^* = f^*g^*$ que é redutível, absurdo, assim temos que $G$ é irredutível em $\F[x,y]$.
\end{proof}

\begin{defn}
	Seja $\X$ uma curva definida sobre $\F_q$, isto é, as equações que a define tem seus coeficientes em $\F_q$. Os pontos de $\X$ que tem todas as coordenadas em $\F_q$ são ditos pontos racionais.
\end{defn}

Seja $V_l$ o espaço vetorial de polinômios de grau total  no máximo $l$,  em duas variáveis $x,y$ e com coeficientes em $\F_q$. Considere $G$ um polinômio como o da proposição \ref{teor:irre_fecho} (em particular, o grau total de $G$ é $m$), $P_1,P_2,\cdots,P_n$ pontos racionais  da curva definida por $G$. Definimos o código $C$ por:
$$C=\{ (f(P_1),f(P_2),\cdots,f(P_n)) ; \, \, f\in V_l\}.$$

\begin{thm}\label{teo:cotarefe}
	Se no código $C$, definido  acima tem-se que $n>lm$, então para sua distância mínima $d$ e sua dimensão $k$ são dadas por:
	$$ d \maiori n-lm;$$
	$$k=\left\lbrace\begin{array}{ll}
				\combinacao{l+2}{2}, & \textrm{ se } l<m;\\
				lm+1 - \combinacao{m-1}{2}, & \textrm{ se } l \maiori m.
	                 \end{array}\right.$$
\end{thm}
\begin{proof}
	Primeiro queremos encontrar a dimensão do espaço vetorial $V_l$ que tem como base o conjunto formado pelo monômios de grau menor ou igual a $l$, conjunto esse que tem cardinalidade igual a $\sum_{i=0}^l (l+1)-i = \dfrac{2(l+1)^2-(l+1)l}{2}\! = \dfrac{(l+1)(l+2)}{2} = \combinacao{l+2}{2}$, assim a dimensão de $V_l$ é $\combinacao{l+2}{2}$.
	
	Seja $F\in V_l$, se $G$ for um fator de $F$ temos que a palavra correspondente a $F$ no código é zero. Agora dada uma palavra nula no código e $F\in V_l$ o polinômio que a gera temos que a curva $\Y$ definida por $F=0$ e $G=0$ tem grau $l'\menori l$ e $l'\menori m$ e mais, temos que $P_1,\cdots,P_n$ estão na interseção de $\Y$ com $\X$. O teorema de Bézout nos garante que se $\Y$ e $\X$ não tem um fator em comum o número de pontos na intersessão é menor ou igual a $l'm \menori lm$, mas por hipótese $n>lm$, logo $F$ tem $G$ como seu fator. Assim temos que as funções em $V_l$ que geram a palavra zero é um subespaço vetorial de dimensão $l-m$ dado  por $GV_{l-m}=\{GH ; \, H \in V_{l-m}\}$. 
	
	Se $l<m$ temos que $V_{l-m} = \emptyset$ logo a dimensão do código é dada por $k=\combinacao{l+2}{2}$. Caso contrário, teremos que $k=\combinacao{l+2}{2} - \combinacao{l-m+2}{2} = lm+1-\combinacao{m-1}{2}.$
	
	Agora queremos demonstrar que a distância mínima do código é $d\maiori n-lm$. De fato, seja $w\in C$ uma palavra não nula, suponha que $w$ tem mais que $lm$ coordenadas nulas. Seja $F\in V_l$ um polinômio que gere $w$, tomemos a curva $\Y$ definida por $F=0$, como $gr(F)\menori l$ temos que $gr(\Y)\menori l$ logo $\#\Y \inter \X \menori lm$, pelo teorema de Bézout. Com um possível reordenamento das coordenadas de $w$ podemos supor que $F(P_1)=F(P_2)=\cdots = F(P_{lm})=\cdots=F(P_j)=0$, pela nossa construção temos que $\{P_1,\cdots,P_j\} \contido \Y \inter \X$. Logo $j\menori lm$, assim $d\maiori n-lm$, como queríamos demonstrar.
\end{proof}

\chapter{Códigos de Avaliação}

Este capítulo está dedicado a construção dos códigos de avaliação.

\section{Funções  Peso, Grau e Ordem}

\begin{defn}

	Seja $R=\F[x_1,\cdots,x_m]$ o anel de polinômios a $m$ variáveis sobre o corpo $\F$, suponha que exista uma ordem total $\ordem$ no conjunto de monômios de $R$ tal que para quaisquer monômios $M_1,\,\,M_2$ e $M$ temos que:
	\begin{enumerate}
		\item \label{ordem:1} Se $M\neq 1$, então $1 \ordem M$;
		\item \label{ordem:2} Se $M_1 \ordem M_2$, então $MM_1\ordem MM_2$.
	\end{enumerate}
	Assim dizemos que $\ordem$ é uma ordem de admissão ou ordem de redução em monô\-mios.
\end{defn}
	No que segue $R$ é uma $\F-$álgebra, ou seja, um anel comutativo com unidade tal que $\F \contido R$ como subanel. E mais, o simbolo $-\infty$ é tal que para todo $n\in \N_0\uniao \{-\infty\}$, $-\infty +n = -\infty$.
\begin{defn}\label{fun:ordem}\index{Função!  ordem} Uma função $\rho : R \emm \N_0\uniao \{-\infty\}$, que satisfaz as propriedades abaixo é chamada de {\bf função  ordem}.
	
	\begin{enumerate}
		\item $\rho(f)=-\infty$ se, e somente se, $f=0$;
		\item $\rho(\lambda f) = \rho(f)$ para todo $\lambda \in \F \setminus \{0\}$;
		\item $\rho(f+g)\menori max\{\rho(f),\rho(g)\}$ e a igualdade é válida quando $\rho(f) \neq \rho(g)$;
		\item Se $\rho(f) < \rho(g)$ e $h\neq 0$, então $\rho(fh)<\rho(gh)$;
		\item Se $\rho(f)=\rho(g)\neq 0$, então existe $\lambda \in \F \setminus \{0\}$ tal que $\rho(f-\lambda g)<\rho(g)$.
		
		Se além dessas propriedades $\rho$ também satisfizer a próxima a chamaremos de função peso\index{Função! peso}.
		\item $\rho(fg) = \rho(f) + \rho(g).$\label{fun:peso}
	\end{enumerate}
\end{defn}

\begin{exemp}
	Um primeiro exemplo de uma função peso é a função grau de polinômios no anel de polinômios em uma variável $\F[x]$.	
\end{exemp}

\begin{defn}
	Uma função grau em\index{Função! grau} $R$ é uma função que satisfaz as propriedades 1, 2, 3, 4 e 6 da definição \ref{fun:ordem}.
\end{defn}

Vejamos agora um resultado que nos traz propriedades para as funções  ordem.
\begin{lem}
	Seja $\rho$ uma função ordem em $R$, então temos:
	\begin{enumerate}
		\item Se $\rho(f)=\rho(g),$ então $\rho(fh)=\rho(gh)$ para todo $h\in R$;
		\item Se $f\in R\setminus \{0\},$ então $\rho(1) \menori \rho(f)$;
		\item $\F = \{f \in R ; \rho(f) \menori \rho(1)\}$;
		\item Se $\rho(f) = \rho(g)$, então existe um único escalar não nulo $\lambda \in \F$ tal que $\rho(f - \lambda g )< \rho(g)$.
	\end{enumerate}
\end{lem}
\begin{proof}
	A demonstração desse lema sairá diretamente da definição de funções ordem.
	
	(1) Seja $\rho(f)=\rho(g)$, temos que existe $\lambda \in \F$ tal que $\rho(f-\lambda g) < \rho(g)$, logo $\rho(fh - \lambda gh) < \rho(gh)$. Podemos escrever $fh=(fh-\lambda gh)+ \lambda gh$, assim $\rho(fh)=\rho(\lambda gh)=\rho(gh).$
	
	(2) Suponha por absurdo que $f\in R$ é um elemento não nulo tal que $\rho(f) < \rho(1)$, então a cadeia $\rho(1)>\rho(f)>\rho(f^2)>\cdots$ é estritamente decrescente, absurdo, pois $\N_0\uniao \{-\infty\}$ é bem ordenado.
	
	(3) Claramente $\F \contido H=\{f\in R ; \rho(f)\menori \rho(1)\}$, agora seja $f\neq 0$ tal que $\rho(f)\menori \rho(1)$ então $\rho(f)=\rho(1)$, assim existe  $\lambda$ tal que $\rho(f-\lambda)<\rho(1)$, logo $f -\lambda = 0$ ou seja $f\in \F$.
	
	(4) Pela definição de função ordem, temos a existência do $\lambda$, falta assim mostrar a unicidade. Suponha que $\lambda, \nu \in \F$ são tais que $\rho(f-\lambda g)< \rho(g)$ e $\rho(f - \nu g)<\rho(g)$. Então temos que $\rho(f-\lambda g -(f- \nu g)) \menori max\{\rho(f-\lambda g) , \rho(f -\nu g)\} < \rho(g)$. Assim temos que $\rho((\lambda - \nu)g) < \rho(g)$ o que implica que $\lambda - \nu = 0$. Assim $\lambda = \nu$.
\end{proof}

Uma primeira conseqüência sobre a estrutura de um anel $R$ com uma função ordem é dada na seguinte proposição:
\begin{prop}\label{pro:dominio}
	Se existe uma função ordem, $\rho$, em $R$, então $R$ é um domínio de integridade.
\end{prop}
\begin{proof}
	Suponha que existam $f, \,g \in R\setminus \{0\}$ tais que $fg=0$, sem perda de generalidade assumiremos que $\rho(f)\menori \rho(g)$, assim $\rho(f^2)\menori \rho(fg)=\rho(0)=-\infty$ logo $\rho(f^2)=-\infty$, isto é, $f^2=0$. Como $f\neq 0 $, temos que $\rho(1)\menori \rho(f)\menori \rho(f^2)$, absurdo. Logo $fg\neq 0$ assim $R$ é um domínio.
\end{proof}
Agora apresentaremos um exemplo com o qual mostraremos que a recíproca da proposição \ref{pro:dominio} é falsa.

\begin{exemp}
	A $\F-$álgebra $R=\F[x,y]/\langle xy-1\rangle$ é um domínio, mas não tem uma função ordem. De fato, denotando por $\barra{x}$ a classe de equivalência  $x+ \langle xy-1\rangle$ e por $\barra{y}$ a classe $y+ \langle xy-1\rangle$. Como $R$ é um domínio temos que $\barra{x}\neq 0 $ e $\barra{y}\neq 0$. Sendo $\rho$ uma função ordem em $R$ temos que $\rho(1)\menori \rho(\barra{x})$ e assim $\rho(\barra y)\menori \rho(\barra {xy})=\rho(1)$, ou seja, $\rho(\barra{y}) = \rho(1)$, analogamente achamos que $\rho(\barra{x})=\rho(1)$. Observe que $R=\F[\barra{x}]+\F[\barra{y}]$ e assim para todo $f \in R$ concluimos que $\rho(f)\menori \rho(1)$, ou seja, $R = \F$, no entanto $\barra{x} \notin \F$, absurdo.
\end{exemp}

A seguir mostraremos que dada uma $\F-$álgebra $R$ com uma função ordem, essa admite uma $\F-$base com ``boas'' propriedades.

\begin{thm}\label{prop:f-algebra}
	Seja $R$ uma $\F-$álgebra com uma função ordem $\rho$, $\F\neq R$. Então:
	\begin{enumerate}
		\item Existe uma $\F-$base, $\{f_i, i \in \N\}$, para $R$ tal que $\rho(f_i)<\rho(f_{i+1})$ para todo $i\in \N$;
		\item Se $f=\sum_{i=1}^m \lambda_i f_i$ com $\lambda_i\in \F$ e $\lambda_m \neq 0$, temos que $\rho(f) = \rho(f_m)$;
		\item Seja $l(i,j):=l$ o inteiro tal que $\rho(f_if_j) = \rho(f_l)$. Assim, $l(i,j)<l(i+1,j)$ para todo $i,j$;
		\item Seja $\rho_i := \rho(f_i)$. Se $\rho$ é uma função peso então $\rho_{l(i,j)} = \rho_i + \rho_j$.
	\end{enumerate}
\end{thm}
\begin{proof}
	(1) Temos que existe $f\in R$ com $f\notin \F$, pois $R\neq \F$, assim $\rho(1)<\rho(f)$ o que nos dá que $\rho(f^n)< \rho(f^{n+1})$ para todo $n\in \N_0$. Mais ainda, o conjunto dos valores de $\rho$ é infinito. Seja $(\rho_i)_{i\in\N}$ a seqüência crescente de inteiros não negativos tais que os $\rho_i's$  são todos os valores da função ordem, isto é, $\rho(R\setminus \{0\})=\{\rho_i;\, i\in \N\} $. Por definição, para todo $i\in \N$ existe um $f_i \in R$ tal que $\rho(f_i)=\rho_i$ assim $\rho(f_i)<\rho(f_{i+1})$. E mais, pela nossa construção, para todo $f\in R\setminus \{0\} $ existe um $f_i$ tal que $\rho(f)=\rho(f_i)$. Observamos que $\rho_1=\rho(1)$. Agora falta mostrar que $B=\{f_i ; i \in \N\}$ é uma base.  Claramente $B$ é um conjunto linearmente independente, mostremos que ele gera $R$. 
	
	Seja $f\in R$ temos que existe um $f_k\in B$ tal que $\rho(f_k)=\rho(f)$ o que nos dá que existe $\lambda_k \in \F$ de modo que, $\rho(f-\lambda_k f_k)<\rho(f_k)$, novamente temos que existe $f_h$ com $h<k$ tal que $\rho(f-\lambda_k f_k)=\rho(f_h)$ e conseqüentemente existe $\lambda_h$ tal que $\rho(f-\lambda_k f_k -\lambda_h f_h)<\rho(f_h)$,  esse processo deve acabar em no máximo $k$ vezes pois temos apenas $k-1$ $\rho_i's$  menores que $\rho_k$.  Assim chegaremos em $\rho(f-\sum_{i=0}^{k-1} \lambda_{k-i} f_{k-1}) < \rho(1)$,  observe que alguns dos $\lambda_i's$ que aparecem aqui podem ser nulos, $\rho(f-\sum_{i=0}^{k-1} \lambda_{k-i} f_{k-1}) = -\infty $ o que nos dá que $f=\sum_{i=0}^{k-1} \lambda_{k-i} f_{k-1}$, assim $B$ é uma $\F-$base de $R$.
	
	(2) Pela existência da base acima e o fato de que $\rho(f+g) = max\{\rho(f),\rho(g)\}$ se $\rho(f)\neq \rho(g)$, tem-se (2).
	
	(3) Como $\rho(f_i)<\rho(f_{i+1})$ temos que $\rho(f_i f_j) < \rho(f_{i+1}f_j)$ logo $l(i,j)<l(i+1,j).$
	
	(4) Sendo $\rho$ uma função peso temos que $\rho(fg)=\rho(f)+\rho(g)$. Assim $\rho_{l(i,j)} = \rho_i + \rho_j$.
\end{proof}

\section{Códigos de Avaliação}

Nessa seção introduziremos o conceito de código de avaliação. Trabalharemos aqui com um corpo finito com $q$ elementos, $\F_q$. Além disso, aqui $R$  é uma $\F_q-$álgebra com função ordem $\rho$ tal que admite uma base, $\{f_i;i\in \N\}$, de modo que $\rho(f_i)<\rho(f_{i+1})$ para todo $i\in \N$.

Queremos  transformar o espaço vetorial $\F_q^n$ em uma álgebra e para isso precisamos definir uma multiplicação de vetores, assim definiremos a multiplicação $*$ como sendo a multiplicação usual de coordenadas, ou seja, dados $a=(a_1,\cdots,a_n),b=(b_1,\cdots,b_n) \in \F_q^n$ temos que $a*b=(a_1b_1,\cdots,a_nb_n)$.
\begin{defn}
	Definimos por $L_l$ o espaço vetorial gerado por $f_1,\cdots,f_l$.
\end{defn}
\begin{defn}
	Chamaremos de {\bf morfismo de $\F_q-$álgebras},\index{Morfismo! de $\F_q-$álgebras} a uma função $\F_q-$li\-near, \mbox{$\varphi:{R}\emm{\F_q^n}$}, tal que $\varphi(fg)=\varphi(f)*\varphi(g)$.
\end{defn}
Agora definiremos um código de avaliação e o seu código dual.

\begin{defn}
	Seja $L_l$ o espaço vetorial como acima definido e $\varphi$ um morfismo entre $\F_q-$álgebras, assim definimos o {\bf código de avaliação\index{Código! de Avaliação} $E_l$} determinado por $\varphi$, como sendo a imagem de $L_l$ por meio da função $\varphi$, ou seja,
	$$E_l=\varphi(L_l)=\gera{\varphi(f_1),\cdots,\varphi(f_l)}.$$
	Denotaremos por $C_l$ o código dual de $E_l$.
\end{defn}
Observe que a seqüência de códigos $(E_l)_{l\in\N}$ é crescente segundo a inclusão, e mais, como o espaço de chegada da função $\varphi$ tem dimensão finita temos que essa seqüência estabiliza num certo $N$.

Trabalharemos aqui somente com os morfismos sobrejetivos. 

\begin{exemp}
Sejam $R=\F_q[x_1,\ldots,x_n]/I$, onde $I$ é um ideal do anel de polinômios $\F_q[x_1,\ldots,x_n]$, $\Pponto=\{P_1,\ldots,P_m\}$ um subconjunto com $m$ pontos de $V(I)$, consideremos a função de avaliação:
	$$\funcao{ev_\Pponto}{R}{\F_q^m}{f+I}{(f(P_1),\ldots,f(P_m))}.$$

Observe que $ev_\Pponto$ está bem definida, pois para $P\in V(I)$ temos que se $f+I=g+I$ então $f-g\in I$, logo, $(f-g)(P)=0$, ou seja, $f(P)=g(P)$, e mais, desse modo definido $ev_\Pponto$ é um morfismo de $\F_q$-álgebras, visto que $fg(P)=f(P)g(P)$, para todos $f,g \in R$ e $P\in \F_q^n$.

\end{exemp}

\begin{lem}\label{lema:sobrej}
	A função de avaliação acima definida é sobrejetiva.
\end{lem}
\begin{proof}
	Sejam $\,P_j\,=\,(a_{j1},\,\ldots,\,a_{jm})$ e $\,$ os $\,$ conjuntos $\,$ auxiliares $A_{il}=\{a_{jl};\,j=1,\ldots,m\}\setminus \{a_{il}\}$. Definimos os polinômios 
	$$G_i=\prod_{l=1}^{n}\prod_{a\in A_{il}}(x_l-a).$$
	Temos que $G_i(P_j)=0$ sempre que $i\neq j$, e mais $G_i(P_i)\neq 0$. Os polinômios $G_i/G_i(P_i)$ quando aplicados em $ev_\Pponto$ chegam na base canônica de $\F_q^n$ assim $ev_\Pponto$ é sobrejetiva.
\end{proof}

\section{A Cota Ordem}

Um dos principais parâmetros de uma código é a sua distância mínima, contudo,na maioria das vezes é difícil de ser calculada, assim  torna-se importante as cotas para essa distância mínima. Nesta seção estamos interessados em encontrar uma cota inferior para a distância mínima do código $C_l$.

Continuaremos a usar as notações da seção anterior, lembrando que $N$ é o menor natural tal que a seqüência de códigos de avaliação $E_l$ estabiliza. Os morfismos utilizados serão sobrejetivos.

Definimos aqui uma  $N\x n$ matriz $H$ com a $i-$ésima linha sendo $h_i=\varphi(f_i)$ onde $f_i$ são os elementos da base construída anteriormente e $\varphi$ é o morfismo de $\F_q-$álgebras utilizado aqui.

\begin{defn}
	Seja $y\in \F_q^n$. Consideremos as sindromes $s_i(y)=\gera{y,h_i}$ e $s_{ij}(y) = \gera{y,(h_i*h_j)}$. Então a matriz $S(y) = (s_{ij}(y) ;\,\,\, 1\menori i,j\menori N)$ é a matriz síndrome\index{Matriz! Síndrome} de $y$.
\end{defn}

\begin{lem}\label{lema:diagonal}
	Seja $y\in \F_q^n$ e $D(y)$ a matriz diagonal com as coordenadas de $y$ em sua diagonal, então
	$$S(y)=HD(y)H^t,$$ 
	e mais
	$$posto(S(y))=w(y).$$
\end{lem}
\begin{proof}
	Temos que $s_{ij}(y)= \gera{y,(h_i*h_j)} = \sum_l y_lh_{il}h_{jl}$, e mais, denotando $C=S(y)$, temos que $C_{ij}= H_{ik}y_kH^t_{kj}=H_{ik}y_kH_{jk}$, logo a igualdade é válida. 
	
	Agora o posto de $D(y)$ é justamente o peso de $y$. Como a função $\varphi$ é sobrejetiva temos que a matriz $H$ tem posto máximo, ou seja, posto de $H$ é $n$. Assim temos que o posto de $S(y)$ é o mesmo posto de $D(y)$ como queriamos.
\end{proof}

\begin{defn}
	Seja $l\in \N_0$, definimos o conjunto $N_l=\{ (i,j)\in \N^2 ;\,\, l(i,j)=l+1  \}$, onde $l(i,j)$ é o número natural definido no item 3 do teorema \ref{prop:f-algebra}. A sua cardinalidade denotaremos por $\nu_l$.
\end{defn}

\begin{lem}\label{lema:ordem}
	Se $t=\nu_l$ e $(i_1,j_1),\ldots,(i_t,j_t)$ é a enumeração dos elementos de $N_l$ em ordem crescente segundo a ordem lexicográfica\footnote{Dados $(a,b),\,(c,d) \in \N^2$, $(a,b)<(c,d)$ se $a<c$ ou $a=c$ e $b<d$.} em $\N^2$ . Então $i_1<i_2<\cdots <i_t$ e $j_t<j_{t-1}<\cdots<j_1$. Além disso, se $y\in C_l \setminus C_{l+1}$ temos que $s_{i_uj_v}(y)=0$ se $u<v$ e $s_{i_uj_v}(y)\neq 0$ se $u=v$.
\end{lem}
\begin{proof}
	Pela ordem da seqüência temos que $i_1\menori i_2 \menori \cdots \menori i_t$, suponha por absurdo que $i_k=i_{k+1}$. Desse modo temos que $j_k<j_{k+1}$ logo
	$$l+1=l(i_k,j_k)<l(i_k,j_{k+1})=l(i_{k+1},j_{k+1})=l+1$$
	absurdo, logo a seqüência é estritamente crescente.
	
	Agora suponha que $j_k \maiori j_{k+1}$, novamente temos que
	$$ l+1 =l(i_k , j_k) \maiori l(i_k , j_{k-1}) > l(i_{k-1},j_{k-1})=l+1,$$
	outro absurdo, assim essa seqüência é decrescente. 
	
	Seja $y\in C_l$, se $u<v$ temos que $l(i_u,j_v)<l(i_v,j_v)=l+1$, assim $f_{i_u}f_{j_v} \in L_l$ e desse modo $h_{i_u}*h_{j_v}\in E_l$, logo $s_{i_uj_v}(y)=\gera{y,h_{i_u}*h_{j_v}}=0$. Do mesmo modo se $u=v$ temos que $l(i_u,j_v)=l+1$ o que nos dá que $h_{i_u}*h_{j_v}\in L_{l+1}\setminus L_l$, e mais, $f_{i_u}f_{j_v}\congruo \mu f_{l+1}\modulo L_l$ para algum $0\neq \mu \in \F_q$. Assim $h_{i_u}*h_{j_v}\congruo \mu h_{l+1} \modulo E_l$, como $y\notin C_l$ temos que $s_{l+1}(y)=\gera{y,h_{l+1}}\neq 0$ pois $\gera{h_i,y}=0$ para $1\menori i \menori l$ e $h_{l+1} \notin C_l.$ Assim $s_{i_uj_v}(y)\neq 0$.
\end{proof}

Observamos que a matriz $m_{uv}=s_{i_uj_v}(y)$ com $1\menori u,v \menori \nu_l$ como do lema acima, é uma matriz quadrada de posto $\nu_l$, e mais, $(m_{uv})$ é uma sub-matriz de $S(y)$, logo o posto de $S(y)$ é maior ou igual a $\nu_l$. Isso juntamente como os lemas \ref{lema:diagonal} e \ref{lema:ordem} demonstram a seguinte proposição:

\begin{prop}\label{prop:peso}
	Se $y\in C_l\setminus C_{l+1}$, então $w(y)\maiori \nu_l$.
\end{prop}

Definiremos agora algumas cotas relacionadas aos códigos de avaliação.

\begin{defn}
	Chamaremos de cota ordem\index{Cota! Ordem} aos números
	 $$ d(l) = min\{\nu_m ; m \maiori l \},$$
	 $$ d_{\varphi}(l)=min\{\nu_m; m \maiori l, \,\, C_m \neq C_{m+1}\}.$$
\end{defn}
\begin{thm}
	Os números $d(l)$ e $d_\varphi(l)$ são cotas inferiores para a distância mínima de $C_l$. Mais ainda, $d(C_l)\maiori  d_\varphi(l) \maiori d(l).$
\end{thm}
\begin{proof}
	Pela proposição \ref{prop:peso} temos que  $d(C_l)\maiori \nu_l\maiori d_\varphi(l).$
\end{proof}

\section{Semigrupos}

Nesta seção falaremos de semigrupos e a sua associação aos códigos de avaliação.

\begin{defn}
	Um semigrupo numérico\index{Semigrupo! numérico} é um subconjunto $\Lambda$ de $\N_0$ com as seguintes propriedades:
	\begin{enumerate}
		\item $\Lambda$ é fechado para adição;
		\item $0\in \Lambda$.
	\end{enumerate}
	Os elementos de $\N_0\setminus \Lambda$ são chamados de lacunas e a quantidade desses elementos será denotada por $g=g(\Lambda)$ (aqui $g$ pode ser infinito).
\end{defn}

Suponha agora $\rho$ é uma função peso em uma $\F_q-$álgebra $R$, por \ref{fun:ordem}-\ref{fun:peso} temos que o conjunto $\Lambda = \{\rho(f) ; \, f\in R,\, f\neq 0\}$ é um semigrupo numérico chamado de {\bf semigrupo de $\rho$} \index{Semigrupo! de $\rho$}. Se $g< \infty$, então existe um $n \in \Lambda$ tal que se $x \in \N$ e $n\menori x$ então $x\in \Lambda$. Ao menor $n$ com essa propriedade chamamos de {\it condutor}\index{Condutor} de $\Lambda$ e denotamos por $c=c(\Lambda)$. Observe que $c-1$ é o maior lacuna de $\Lambda$ desde que $g>0$.

\begin{defn}
	Seja $(\rho_l)_{l\in \N}$ uma enumeração do semigrupo $\Lambda$ tal que $\rho_l<\rho_{l+1}$ para todo $l$. Denotamos por $g(l)$ o número de lacunas menores que $\rho_l$.
\end{defn}

\begin{lem}\label{lema:semi_fini}
	Sejam $\Lambda$ um semigrupo com finitos lacunas e $l\in \N$, então:
	\begin{enumerate}
		\item $g(l)=\rho_l-l+1$;
		\item $\rho_l\menori l+g-1$, valendo a igualdade se, e somente se, $\rho_l\maiori c$;
		\item Se $l> c-g$, então $\rho_l=l+g-1$;
		\item Se $l\menori c-g$, então $\rho_l < c-1$.
	\end{enumerate}
\end{lem}
\begin{proof}
	(1) O elemento $\rho_l\in \Lambda$ é o $(\rho_l+1)-$ésimo elemento de $\N_0$ e mais, ele é o $(\rho_l+1-g(l))-$ésimo elemento de $\Lambda$. Assim $l=\rho_l+1-g(l)$, ou seja, $g(l)=\rho_l +1 -l$.
	
	(2) Se $g(l)\menori g$ temos que $\rho_l+1-l \menori g$, logo $\rho_l \menori g-1+l$, se $\rho_l \maiori c$ temos que todos os lacunas são menores que $\rho_l$ logo $g(l)=g$ valendo assim a igualdade.
	
	(3) Temos que $c$ é o $(c+1)-$ésimo elemento de $\N_0$ e mais, é o $(c+1-g)-$ésimo termo de $\Lambda$m, assim $c=\rho_{c+1-g}$. Tomando $l>c-g$ teremos $\rho_l \maiori \rho_{c+1-g}=c$ e mais, conseguimos assim $\rho_l=g-1+l$.
	
	(4) Seja $l \menori c-g$, assim $\rho_l\menori l+g-1\menori c-1$. Contudo $c-1$ é um lacuna ou é negativo. Para $c=0$ não tem sentido a proposição, então temos que $c-1$ é um lacuna, assim $\rho_l < c-1$.
\end{proof}

A próxima proposição nos traz um resultado sobre o condutor de um semigrupo numérico que vai influenciar em uma importante definição.
\begin{prop}
	Seja $\Lambda$ um semigrupo numérico com número de lacunas $g<\infty$, então $c\menori 2g$ e a igualdade é válida se, e somente se, para todo lacuna $s$ temos que $c-1-s$ não é um lacuna.
\end{prop}
\begin{proof}
	Observemos que os pares $(s,t)\in \N^2_0$ tais que $s+t=c-1$, pelo fato de $c-1$ ser um lacuna e de $\Lambda$ ser fechado em relação a soma, tem pelo menos  um dos dois termos de cada par como um lacuna. Como temos $c$ pares desses, levando em consideração a ordem, temos que existem pelo menos $\left[\dfrac {c+1}{2} \right]$ lacunas, logo $c\menori 2g$.
	
	Agora se a igualdade é válida, temos que $g=\dfrac c2$ assim dados $s,t \in \N_0$ tais que $s+t=c-1$ temos que apenas um dos dois ($s$ ou $t$) é um lacuna, logo sendo $s$ um lacuna $c-1-s$ não pode ser lacuna.
	
	Supondo que se $s$ for um lacuna temos que $c-1-s$ não é, isso nos dá que apenas um dos termos dos pares $(s,t)$ tais que $s+t=c-1$ é um lacuna, assim temos exatamente $\dfrac c2$ lacunas.
\end{proof}
O resultado anterior justifica a seguinte definição:
\begin{defn}
	Um semigrupo numérico é chamado simétrico\index{Semigrupo! Simétrico} se $c=2g$.
\end{defn}
\begin{defn}
	Dizemos que um semigrupo numérico é finitamente gerado se existe um conjunto $A=\{a_1,\ldots,a_k\}\contido \Lambda$ tal que dado $\lambda \in \Lambda$ temos que existem $x_1,\ldots,x_k \in \N_0$, tais que $\lambda =\sum_{i=1}^k x_ia_i$. Assim falamos que $A$ gera $\Lambda$ e escrevemos $\Lambda=\gera{A}$.
\end{defn}
A respeito de subgrupos numéricos finitamente gerados o primeiro resultado que apresentamos é:
\begin{prop}\label{prop:gerado}
	Sejam $a,b\in \N$ tais que $mdc(a,b)=1$. O semigrupo gerado por $a$ e $b$ é simétrico, tem como último lacuna o número $ab-a-b$, como seu condutor o número $(a-1)(b-1)$ e o número total de lacunas é $(a-1)(b-1)/2$.
\end{prop}
\begin{proof}
	Como $mdc(a,b)=1$, temos que todo inteiro $m$ pode ser escrito como $m=xa+yb$, de maneira única com $0\menori y <b$.
	
	Pelo fato acima observamos que o maior lacuna possível é $(b-1)a -b$ e de fato esse número é um lacuna, pois não é possível escreve-lo como  $(b-1)a-b= xa +yb$  com $x,y \in \N_0$, assim temos que o condutor é $c=(b-1)a-b+1 = (a-1)(b-1)$.
	
	Agora mostremos que o semigrupo é simétrico. Suponhamos por absurdo que o semigrupo não seja simétrico, ou seja, existe $s,t\in \N$ lacunas tais que $s+t=c-1$ onde $c$ é o condutor. Podemos escrever $s=x_1a+y_1b$ e $t=x_2a+y_2b$, assim temos que $c-1=ab-a-b=(x_1+x_2)a+(y_1+y_2)b$. Observe que $0\menori x_1+x_2 \menori 2b-2$ e mais $y_1+y_2\menori -2$. Disto segue que:
	$$D=(-y_1-y_2-1)b=(x_1+x_2-b+1)a\implica$$
	$$0<b\menori(-y_1-y_2-1)b=(x_1+x_2-b+1)a \menori (b-1)a<ba$$
	Como $mdc(a,b)=1$ temos que $a|(-y_1-y_2-1)$ e que $b|(x_1+x_2-b+1)$, assim chegamos que $0<\dfrac{D}{ab}<1$, absurdo, logo $\Lambda$ é simétrico.
	
	Como o semigrupo é simétrico temos que $c=2g$, logo $g=(a-1)(b-1)/2$.
\end{proof}

Faremos agora mais um lema técnico sobre semigrupos.

\begin{lem}\label{lema:cardi}
	Sejam $\Lambda$ um semigrupo numérico com finitos lacunas e $s\in \Lambda$. Então temos que $\#(\Lambda \setminus \{s+\lambda;\,\lambda \in \Lambda\}) = s$.
\end{lem}
\begin{proof}
	Seja $c$ o condutor de $\Lambda$, $T=\{t\in \N_0;\,\, t\maiori s+c\}$, claramente temos que $T\contido \Lambda$, e mais, $T\contido s+\Lambda=\{s+\lambda;\,\lambda \in \Lambda\}$. Seja $U=\{u\in \Lambda; u<s+c\}$, temos que $\#U=s+c-g$, além disso $\Lambda =  U\uniao T$. Seja $V=\{v\in s+\Lambda;\,\, s\menori v<s+c\}$, temos que $\#V=s+c-g-s=c-g$, e mais, $s+\Lambda = V \uniao T$. Observe que as uniões acima são disjuntas e mais, $V\contido U$. Assim temos que:
	$$\#(\Lambda\setminus s+\Lambda)=\#(U\uniao T \setminus V\uniao T) = \#(U\setminus T) = s+c-g -(c-g) = s.$$
	Como queríamos demonstrar.
\end{proof}
Uma conseqüência quase imediata desse lema é:

\begin{prop}\label{prop:dim}
	Seja $f$ um elemento não nulo de uma $\F_q-$álgebra $R$ com uma função peso $\rho$. Então $dim_{\F_q}(R/\gera{f})=\rho(f)$.
\end{prop}
\begin{proof}
	Sejam $\Lambda$ o semigrupo  da função peso $\rho$ e  $s=\rho(f)$. Tomemos a seqüência $(\rho_i)_{i\in \N} $ dos elementos de $\Lambda$ em ordem crescente. Pela propriedade \ref{fun:ordem}-\ref{fun:peso} temos que a imagem dos elementos não nulos do ideal $\gera{f}$ segundo a função $\rho$ é o conjunto $s+\Lambda$. Como feito antes, para todo $\rho_i \in \Lambda$, existe um $f_i\in R$ tal que $\rho(f_i)=\rho_i$ e caso $\rho_i\in s+\Lambda$ podemos tomar $f_i\in \gera{f}$. Os conjuntos $\{f_i; \,i\in \N\}$ e $\{f_i; \, i\in \N,\, \rho_i \in s+\Lambda\}$ formam uma base para a álgebra $R$ e o ideal $\gera{f}$ respectivamente, vide demonstração de \ref{prop:f-algebra}. Desse modo as classes de equivalência $f_i$ módulo $\gera{f}$ com $i\in \N$ e $\rho_i\in \Lambda \setminus (s+\Lambda)$ formam uma base para o quociente $R/\gera{f}$. Assim a sua dimensão é a cardinalidade da base que é $s$ pelo lema \ref{lema:cardi}, ou seja, $\rho(f)$.
\end{proof}

\begin{defn}
	Seja $R=\F_q[x_1,\ldots,x_n]/I$, onde $I$ é um ideal do anel de polinômios $\F_q[x_1,\ldots,x_n]$, para $f+I \in R$ dizemos que $P\in \F^n_q$ é um zero de $f+I$, se $P\in V(I)$ e $f(P)=0$.
\end{defn}

\begin{lem}\label{lema:numer_zeros}
	Seja $R$ uma $\F_q-$álgebra finita com uma função peso $\rho$. Seja $f\in R$ um elemento não nulo. Então o número de zeros de $f$ é no máximo $\rho(f)$.
\end{lem}
\begin{proof}
	Seja $\Pponto$ o conjunto de zeros de $f$ e $t=\#\Pponto$. A função de avaliação, $ev_\Pponto :R\emm \F_q^t$, é uma função linear e  pelo lema \ref{lema:sobrej} temos que $ev_\Pponto$ é sobrejetiva. Isso nos garante que $R/Ker(ev_\Pponto) \iso \F_q^t$. Observe que $\gera{f} \contido Ker(ev_\Pponto)$ e olhando ambos como sub-espaço vetorial de $R$ temos que $dim_{\F_q}(\gera{f})\menori dim_{\F_q}(Ker(ev_\Pponto))$. Assim segue que $t=dim_{\F_q}(R/Ker(ev_\Pponto))=dim_{\F_q}(R)-dim_{\F_q}(Ker(ev_\Pponto)\menori dim_{\F_q}(R)-dim_{\F_q}(\gera{f}) = \dim_{\F_q}(R/\gera{f}) = \rho(f)$ por \ref{prop:dim}.
\end{proof}

\section{Código de Avaliação Via Semigrupos}

Aqui estamos interessados em  encontrar uma cota para a distância mínima dos códigos de avaliação $E_l$.

Nessa seção iremos supor que $\rho$ é uma função peso em $R=\F_q[x_1,\ldots,x_m]/I$, onde $I$ é um ideal de $\F_q[x_1,\ldots,x_m]$ (anel de polinômios em $m$ variáveis). Seja $(\rho_i)_{i\in \N}$ a enumeração do semigrupo de $\rho$ em ordem crescente. Tomemos $\Pponto$ como um conjunto com $n$ pontos do conjunto $V(I)$ (variedade algébrica gerada por $I$). A função de avaliação $ev_\Pponto : R\emm \F_q^n$ nos define os códigos de avaliação
$$E_l=\{ev_\Pponto(f); \,\, f\in R ,\, \rho(f)\menori \rho_l\}.$$

\begin{thm}\label{teo:dista}
	A distância mínima do código $E_l$ é maior ou igual a $n-\rho_l$. Se $\rho_l<n$, temos que $dim_{\F_q}(E_l)=l.$
\end{thm}
\begin{proof}
	Seja $c$ uma palavra código não nula de $E_l$, então existe $f\in R\setminus \{0\}$ tal que $\rho(f)\menori \rho_i$ e $c=ev_\Pponto(f)$. Temos que $c_i=f(P_i)$ para todo $0<i<n+1$, onde os $c_i$'s são as coordenadas da palavra $c$, pelo lema \ref{lema:numer_zeros} temos que o número de zeros de $f$ é no máximo $\rho(f)\menori \rho_l$, assim $w(c)\maiori n-\rho_l$.
	
	Agora suponha que $\rho_l<n$. Temos que $E_l$ é a imagem do espaço vetorial $L_l$ através da função $ev_\Pponto$. Se $f\in L_l$ e $ev_\Pponto(f)=0$ então $f$ tem pelo menos $n$ zeros, mas pelo lema \ref{lema:numer_zeros} se $f$ é não nulo então $f$ admite no máximo $\rho(f)$ zeros, contudo $\rho(f)\menori \rho_l <n$, assim $f$ tem de ser nulo, concluímos assim que $Ker(ev_\Pponto |_{L_l})=\{0\}$, o que nos dá que $dim_{\F_q}(E_l)=dim_{\F_q}(L_l)$, lembremos que $L_l$ é o espaço vetorial gerado pelo conjunto $\{f_1,\ldots,f_l\}$, logo sua dimensão é $l$, assim $dim_{\F_q}(E_l)=l$.
\end{proof}
\begin{cor}
	Seja $\rho$ uma função peso com $g$ lacunas. Se $\rho_k<n$, então $E_k$ é um $[n,k,d]-$código tal que $d\maiori n+1-k-g$.
\end{cor}
\begin{proof}
	Pelo teorema \ref{teo:dista} temos que $d\maiori n-\rho_k$. Agora pelo lema \ref{lema:semi_fini} temos que $\rho_k \menori k+g-1$, assim segue que $d\maiori n+1-k-g.$
\end{proof}


\chapter{Exemplos}
Este capítulo está dedicado a apresentação de alguns exemplos dos códigos que nesse texto foram definidos.
As notações aqui  utilizadas são na maioria as mesmas do  capítulo 3, ou seja, $R$ denotará sempre uma $\F_q-$álgebra, $\rho$ uma função peso, $l(i,j)$, o inteiro tal que $\rho(f_if_j)=\rho_{l(i,j)}$, onde $\{f_1,f_2,\ldots\}$ é uma $\F_q$-base para $R$, etc.

\section{Um Primeiro Exemplo}

Esse primeiro teorema do capítulo vai nos permitir garantir a existência de funções ordem e peso em determinadas condições.

\begin{thm}\label{teo:exemp1}
    Sejam $R$ uma $\F-$álgebra e $\{f_1,f_2,\ldots\}$ uma $\F-$base do $\F-$espaço vetorial $R$, com $f_1=1$. Sejam, ainda, $(\rho_i)_{i\in \N}$ uma seq{\"u}ência estritamente crescente de inteiros não negativos e $\rho : R \emm \N_0 \uniao \{-\infty\}$ a função definida por $\rho(0)=-\infty$ e $\rho(f)=\rho_i$ se $f\neq 0$ e $i$ for o menor inteiro tal que $f\in L_i$, onde $L_i$ é o $\F-$subespaço vetorial gerado por $\{f_1,\ldots,f_i\}$.  Se para todo $(i,j)\in \N^2$ tem-se que $l(i,j)<l(i+1,j)$, então $\rho$ é uma função ordem e, mais ainda, se $\rho_{l(i,j)}=\rho_i+\rho_j$, então $\rho$ é uma função peso.

\end{thm}
\begin{proof}
    Diretamente da definição da função $\rho$, vemos que ela satisfaz as condições 1, 2, 3 e 5 para ser uma função peso. Mostraremos que ela também satisfaz as outras duas condições (4 e 6).

    Para $f\in R\setminus \{0\}$, associamos o número $\iota(f)$, o qual é o menor inteiro positivo tal que $f\in L_{\iota(f)}$. Dados $f,g \in R\setminus \{0\}$ temos que
    $$f=\sum_{i\menori \iota(f)}\lambda_i f_i, \, g=\sum_{i\menori \iota(g)}\nu_i f_i \, fg=\sum_{i\menori \iota(fg)}\mu_i f_i\textrm{ e } f_if_j=\sum_{l\menori l(i,j)}\eta_{ijl}f_l.$$
    com $\lambda_{\iota(f)}\neq 0,\,\, \nu_{\iota(g)}\neq 0,\,\, \mu_{\iota(fg)}\neq 0$ e $\eta_{ijl(i,j)}\neq 0 $.

    Observe que
    $$fg=\left( \sum_{i\menori \iota(f)}\lambda_i f_i\right) \left( \sum_{j\menori \iota(g)}\nu_j f_j\right) = \sum_{i\menori \iota(f)}\sum_{j\menori \iota(g)}\lambda_i\nu_jf_if_j =$$
    $$\sum_{i\menori \iota(f)}\sum_{j\menori \iota(g)}\lambda_i\nu_j\sum_{l\menori l(i,j)}\eta_{ijl}f_l = \sum_{i\menori \iota(f)}\sum_{j\menori \iota(g)}\sum_{l\menori l(i,j)}\lambda_i\nu_j\eta_{ijl}f_l,$$
    assim temos que
    $$\mu_l = \sum_{l(i,j)=l}\lambda_i\nu_j\eta_l f_l.$$

    Por hipótese temos que $l(i,j)<l(i+1,j)$, assim $l(i,j)<l(\iota(f),\iota(g))$ se $i<\iota(f)$ ou $j<\iota(g)$. Supondo $i=\iota(f)$ e $j=\iota(g)$, temos que $\mu_{\iota{fg}}=\lambda_i\nu_j\eta_{ijl(i,j)} \neq 0$, o que nos garante que $\iota(fg)=l(\iota(f),\iota(g))$.

    Agora dados $f,g,h \in R\setminus \{0\}$, com $\rho(f)<\rho(g)$,   temos que $\rho(fh)=\rho_{\iota(fh)}=\rho_{l(\iota(f),\iota(h))}<\rho_{l(\iota(g),\iota(h))}=\rho_{\iota(gh)}=\rho(gh)$, logo a condição 4 é verificada e $\rho$ é uma função ordem.

    Agora assumindo que $\rho_{l(i,j)}=\rho_l+\rho_j$, temos que $\rho(fg)=\rho_{\iota(fg)}=\rho_{l(\iota(f),\iota(g))}=\rho_{\iota(f)}+\rho_{\iota(g)} = \rho(f)+\rho(g).$ Ou seja, a condição 6 também é satisfeita sendo então $\rho$ uma função peso.
\end{proof}

\begin{exemp}\label{ex:exem1}
    Seja $\X$ a curva definida pelo polinômio $P(x,y)\in \F_q[x,y]$, $P(x,y)=x^m + y^{m-1} + G(x,y)$, com $gr_{\tnormal{ total}}(G(x,y))<m-1$. Como $P(x,y)$ é um polinômio irredutível, temos que o anel $R=\F_q[x,y]/\gera{P(x,y)}$ é um domínio de integridade e mais, $R$ é uma $\F_q-$álgebra. Assim, construído $R$ e denotando por $\ba{x},\, \ba{y}$ as classes de equivalências $ x + \gera{P(x,y)}$ e $y+\gera{P(x,y)}$, respectivamente, temos que $R$ admite uma função peso, $\rho$, tal que $\rho(\ba{x})=m-1$ e $\rho(\ba{y})=m$.

\end{exemp}

\begin{proof}
    O conjunto $B=\{\ba{x}^{\alpha}\ba{y}^{\beta}\in R;\, \alpha < m\}$, é uma $\F_q-$base para $R$. De fato, primeiro observemos que $\ba{x}^m = -\ba{y}^{m-1} - \ba{G}$, e dado $\ba{h}\in R$ temos que $\ba{h}=h(x,y)+\gera{P(x,y)}$ e mais, $h(x,y)=\sum_{i=0}^{a}\sum_{j=0}^b \lambda_{ij}x^iy^j$, com $\lambda_{ab}\neq 0$. Para demonstrar que $B$ gera $R$ vamos supor sem perda de generalidade que $h(x,y)$ é um monômio, visto que no máximo ele é uma soma de monômios, assim $h(x,y)=x^ay^b$, se $a<m$ temos que $h(x,y) \in B$ assim não temos o que fazer, suponha então $a\maiori m$.

    Podemos escrever $a=km+c$ com $c<a$, assim $\ba{x}^a\ba{y}^b = \ba{x}^c\ba{y}^b(-\ba{y}^{m-1} - \ba{g})^k$, temos que $\ba{x}^c\ba{y}^{b+k(m-1)} \in B$ e os demais monômios de $\ba{x}^c\ba{y}^b(-\ba{y}^{m-1} - \ba{g})^k$ tem grau em $x$ menor que $a$, logo repetindo esse processo recursivamente para os demais temos que $\ba{h}$ é escrito como combinação linear de elementos de $B$.

    Agora queremos mostrar que $B$ é um conjunto linearmente independente.
    Seja $\sum \lambda_{ij}\ba{x}^i \ba{y}^j = 0$ uma soma finita de elementos de $B$, equivalentemente mostraremos para $\sum \lambda_{ij} x^iy^y = f(x,y)P(x,y)$ para algum $f(x,y)\in \F_q[x,y]$, temos que se $f(x,y)\neq 0$, $gr_x(f(x,y)P(x,y))>m$ e mais, $gr_x\left( \sum \lambda_{ij} x^iy^y \right)<m$ logo a igualdade não é válida, sendo $f(x,y)=0$ temos que os $\lambda_{ij}'$s são nulos, pois esses monômios são linearmente independentes em $\F_q[x,y]$ assim $B$ é uma $\F_q-$base para $R$.

    Seja $\{f_1,f_2,f_3,\ldots\}$ uma enumeração do conjunto $B$. Para $f_i=\ba{x}^{\alpha}\ba{y}^{\beta}$ definimos $\rho_i=\alpha(m-1) + \beta m$. Sendo $D=\{(a,b)\in \N_0^2 ; a<m\}$, a função $\varphi:D\emm \N_0$, tal que $\varphi(a,b)=a(m-1)+bm$, é injetiva, já que $mdc(m,m-1)=1$ e $a<m$, assim se $i\neq j$ temos que $\rho_i \neq \rho_j$.

    Reordenando se preciso, assumiremos que a seq{\"u}ência $(\rho_i)_{i\in \N}$ é estritamente crescente.

    Seja $f_i=\ba{x}^{\alpha}\ba{y}^{\beta}$ e $f_j=\ba{x}^{\gamma}\ba{y}^{\delta}$, com $\alpha<m$ e $\gamma<m$. Seguindo as notações do teorema anterior, queremos mostrar que $l(i,j)<l(i+1,j)$ e para isso vamos mostrar que $\rho_{l(i,j)}=\rho_i+\rho_j$. Separaremos em 2 casos, $\alpha+\gamma <m$ e $m\menori \alpha + \gamma <2m$.
    \begin{enumerate}
        \item Se $\alpha+\gamma <m$, temos que $f_if_j \in B$, logo $f_if_j=f_{l(i,j)}$, onde $l(i,j)$ é o menor inteiro $l$, tal que $f_if_l\in L_l=\gera{f_1,\ldots,f_l}$, assim $\rho_l(i,j) = \rho_i+\rho_j$.

        \item Agora vamos supor $\alpha + \gamma \maiori m$, assim $\alpha + \gamma = m+\epsilon$ com $0\menori \epsilon < m$. Tomamos $n=\beta+\delta$, temos que
        $$ f_if_j = (\ba{x}^{\alpha}\ba{y}^{\beta})(\ba{x}^{\gamma}\ba{y}^{\delta}) = \ba{x}^{(m+\epsilon)}\ba{y}^{(\beta+\delta)}=\ba{x}^{\epsilon}\ba{y}^n(-\ba{y}^{m-1} - \ba{g})=$$
        $$ -\ba{x}^{\epsilon}\ba{y}^{(n+m-1)} - \ba{x}^{\epsilon}\ba{y}^{n}\ba{g} .$$
        Observe que $\ba{x}^{\epsilon}\ba{y}^{(n+m-1)}\in B$ e mais, se tomarmos $f_l=\ba{x}^{\epsilon}\ba{y}^{(n+m-1)}$ temos que
        $$ \rho_i+\rho_j=(\alpha+\gamma)(m-1)+(\beta+\delta)m = (m+\epsilon)(m-1) + nm =$$
        $$ \epsilon(m-1)+(m-1+n)m = \rho_l.$$

        Um monômio de $G$ com coeficiente não nulo é da forma $x^{\kappa}y^{\lambda}$  com $\kappa \menori gr_{x}(G)=d$ e $\kappa + \lambda < m-1$.

        Se $(\epsilon,\eta),(\kappa,\lambda)\in \N_0^2, \, \epsilon < m, \, \kappa \menori d, \, \kappa+\lambda < m-1$ e $\rho_l = \epsilon(m-1)+(m-1+n)m$, então $\ba{x}^{\epsilon+\kappa}\ba{y}^{\eta+\lambda}\in L_{l-1}$. De fato, mostraremos isso nos dois casos que seguem.
        \begin{enumerate}
            \item  Se $\epsilon + \kappa<m$, então $\ba{x}^{\epsilon+\kappa}\ba{y}^{\eta+\lambda}\in B$ e mais, $(\eta+\lambda)m+(\epsilon+\kappa)(m-1) < \epsilon(m-1)+(\eta + m-1)m = \rho_l$, assim, $\ba{x}^{\epsilon+\kappa}\ba{y}^{\eta+\lambda}\in L_{l-1}$.

            \item Agora, se $\epsilon + \kappa \maiori m$, temos que $\epsilon + \kappa = m +\epsilon'$ com $\epsilon'<\epsilon$, pois $\kappa \menori d <m$ e $\epsilon < m$. Do mesmo modo fazemos $\eta+\lambda = \eta'$. Assim
            $$ \ba{x}^{\epsilon+\kappa}\ba{y}^{\eta+\lambda} = \ba{x}^{m+\epsilon'}\ba{y}^{\eta'}= -\ba{x}^{\epsilon'}\ba{y}^{m-1+\eta'} - \ba{x}^{\epsilon'}\ba{y}^{\eta'}\ba{g}.$$

            Observe que $\rho_{l'}=\epsilon'(m-1)+(m-1+\eta')m = (m+\epsilon')(m-1)+\eta'm = (\epsilon + \kappa)(m-1) + (\eta+\lambda)m < \rho_l$, assim $f_l\in L_{l-1}$, e mais, recursivamente podemos mostrar que $\ba{x}^{\epsilon'}\ba{y}^{\eta'}\ba{g}\in L_{l-1}$.
        \end{enumerate}
    \end{enumerate}

    Assim mostramos que se $f_if_j \in L_{l(i,j)}$ então $\rho_{l(i,j)}=\rho_i+\rho_j$.

    Desse modo  pelo   teorema \ref{teo:exemp1}, temos que $R$ admite uma função peso e que essa do modo que foi construida é gerada por $m-1$ e $m$.
\end{proof}

\begin{exemp}
    Seja $\X$ a curva definida no exemplo \ref{ex:exem1}. O semigrupo numérico gerado pela função peso $\rho$, do mesmo exemplo, tem $g= {{m-1}\choose {2}}$ lacunas. Sejam $Q$ um conjunto de $n$ pontos racionais de  $\X$ e $k$ tal que $\rho_k=lm$, com $(l>m, \textrm{ e } lm<n)$. O código de avaliação $C=E_k$, determinado por $ev_Q$,  é um $[n,k,d]-$código com, $d\maiori n-lm$  e $k=lm+1-g= lm+1 -{{m-1}\choose {2}}$.
\end{exemp}
\begin{proof}
    Temos que o semigrupo determinado $\rho$ é gerado por $m$ e $m-1$. Assim, pela proposição \ref{prop:gerado}, temos que $g={m-1 \choose 2}$. Pelo teorema \ref{teo:dista}, temos que $d\maiori n-lm$. Agora, do lema \ref{lema:semi_fini}(1) segue que $\rho_k=k+g-1$, assim $k=lm-g+1$.
\end{proof}

Note que a dimensão e a cota inferior para a distância mínima do código do último exemplo e do código do teorema \ref{teo:cotarefe} são iguais.


Em vários artigos encontramos referências aos códigos geométricos de Goppa pontuais, códigos da forma $C(D,mQ)$, onde $Q$ é um ponto racional  e $m$ um inteiro. Os códigos de avaliação que aqui construimos foram propostos como um modo de estudo dos códigos de Goppa pontuais de modo simples.
A principio pensava-se que os códigos de avaliação generalisavam os de Goppa pontuais, mas recentemente fora provado que isso é falso.
Nos exemplos que seguem fazemos uma associação dos códigos geométricos de Goppa pontuais  aos de avaliação.

No livro Algebraic Function Fields and Codes  de  Hennin Stichtenoth,  página 113, temos um exemplo de um corpo de funções no qual o gênero é dado por, $g=(m-1)/2$ caso $m$ seja ímpar e $g=(m-2)/2$ caso contrário. 

\begin{exemp}\label{exe:gop1}
    Sejam $K$ um corpo finito de caracteristica diferente de 2, $\X$ a curva definida por $y^2 = f(x) = p_1(x)\cdots p_s(x)\in K[x]$, onde $p_1(x),\ldots,p_s(x)$ são polinômios mônicos irredutiveis distintos entre si, $s\maiori 1$ e $F$ o corpo de frações do anel de coordenada de $\X$. Assim $K$ é o corpo de constantes de $F$ e se $m=gr(f(x))$ for impar temos que o gênero de $F$ é $(m-1)/2$.

    Agora sejam $P,P_1,\ldots,P_n$, places de grau 1 dois a dois disjuntos, $D=P_1+\cdots+P_n$,  $G=lmP$ um divisor em $F/K$, $m<gr(G)=lm<n$  e $supp(G)\inter supp(D)=\emptyset$. Então o código geométrico de Goppa $C(D,G)$ (definição \ref{cod:goppa}) tem parâmetros, $k=lm+1-(m-1)/2$ e $d \maiori n-lm$.
\end{exemp}

Faremos agora a construção de um código de avaliação.

\begin{exemp}\label{ex:eva2}
    Sejam $K$ um corpo finito com caracteristica diferente de 2, $f(x,y)=y^2+g(x)$ um polinômio em $K[x,y]$ tal que $g(x)\in K[x]$, $\X$ a curva plana gerada por $f(x,y)$ e $m=gr(g(x))$ impar, em particular $f(x,y)$ é irredutível. Faça $R=K[x,y]/\gera{f(x,y)}$.  Assim $R$ é uma $K-$álgebra que admite uma função peso, $\rho$, gerada por 2 e $m$.
\end{exemp}

\begin{proof}
    Denotaremos por $\ba{h}$ à classe $h+\gera{f(x,y)}$ de $R$. O conjunto $B=\{\ba{x}^b\ba{y}^a ; a<2\}$ é uma $K-$base para $R$. De fato, observe que dada uma soma finita da forma $\sum \lambda_{ab}\ba{x}^b\ba{y}^a = 0$, teriamos em $K[x,y]$ a igualdade $\sum \lambda_{ab}x^by^a = h(x,y)f(x,y)$, para algum $h(x,y)\in K[x,y]$, contudo $gr_y(h(x,y)f(x,y)) \maiori 2$ enquanto $gr_y\left(\sum \lambda_{ab}x^by^a \right) < 2$, assim $B$ é um conjunto linearmente independente. Mostremos agora que $B$ gera $R$.

    É suficiente mostrar que $B$ gera os elementos da forma $\ba{y}^c\ba{x}^d$, pois todos elementos em $R$ são somas de elementos desse tipo. Se $c<2$ temos que $\ba{y}^c\ba{x}^d\in B$, ou seja, não tem o que fazer, vamos supor agora que $c\maiori 2$.

    Podemos então escrever $c=a+2k$ com $a<2$, assim $\ba{y}^c\ba{x}^d=\ba{y}^a\ba{x}^d\ba{g(x)}^k$, pois $\ba{y}^2=\ba{g(x)}$, logo $B$ é uma $K-$base de $R$.

    Enumeramos $B$ como $\{f_1,f_2,\ldots\}$. Agora sendo $f_i=\ba{y}^a\ba{x}^b \in B$, definimos $\rho_i=2b+am$, observe que $2b+am = 2\tilde b +\tilde a m$, com $a,\tilde a<2$ se, e somente, $a=\tilde a$ e $b=\tilde b$, visto que $mdc(2,m)=1$, assim reenumerando caso necessário, assumiremos que $(\rho_i)_{i\in \N}$ é uma seq{\"u}ência estritamente crescente.

    Seja $l(i,j)$ o menor inteiro tal que $f_if_j \in L_{l(i,j)}$, queremos mostrar que $l(i+1,j)>l(i,j)$, ou seja, temos que mostrar que $\rho_{l(i,j)}<\rho_{l(i+1,j)}$, para isso provaremnos que $\rho_{l(i,j)}=\rho_i+\rho_j$.

    Sejam $f_i=\ba{x}^a\ba{y}^b$ e $f_j=\ba{x}^c\ba{y}^d$, com $b<2$ e $d<2$. Se $b+d<2$ temos que $f_if_j\in B$ e logo  $\rho_{l(i,j)}=\rho_i+\rho_j$. Suponha agora  $b+d\maiori 2$, temos que $b+d=2+\lambda$ com $\lambda<2$, assim $f_if_j=\ba{y}^{\lambda}\ba{x}^{a+c}\ba{g(x)}$. Note que para $f_h=\ba{y}^{\lambda}\ba{x}^{a+c+m}$, temos que $\rho_h=2a+2c+2m+\lambda m = 2(a+c) +m(\lambda +2) = 2(a+c) + m(b+d) = \rho_i + \rho_j$, observe também que outro monômio $f_t$ de $\ba{y}^{\lambda}\ba{x}^{a+c}\ba{g(x)}$ é tal que $\rho_t < \rho_h$ assim $l(i,j)=h$, concluimos então que $\rho_{l(i,j)} = \rho_i + \rho_j$ e assim $\rho_{l(i+1,j)}=\rho_{i+1}+\rho_j > \rho_i+\rho_j = \rho_{l(i,j)}$.

    Agora pelo teorema \ref{teo:exemp1}, temos que existe uma função peso em $R$.
\end{proof}

\begin{exemp}
    Com as hipóteses do exemplo \ref{ex:eva2}, tome um conjunto $P$ com $n$ pontos racionais distintos de $\X$. Além disso, seja $k$ dado por $\rho_k=lm$. Temos então que o código $E_k$ tem parâmetros iguais ao do código $C(D,G)$ visto em \ref{exe:gop1}.
\end{exemp}

	\newpage
 \bibliographystyle{plainnat}

	\addcontentsline{toc}{chapter}{Índice Remissivo}
	\printindex
\end{document}